\documentclass[a4paper,11pt]{article}
\usepackage{fullpage}
\usepackage[T1]{fontenc}
\usepackage[utf8]{inputenc}
\usepackage[font=small]{caption}
\usepackage{latexsym,amssymb,amsmath,amsthm,graphicx,enumerate,mathtools}
\usepackage{cite}
\usepackage{braket}
\usepackage{stmaryrd}
\usepackage[most]{tcolorbox}
\usepackage{tikz}
\usetikzlibrary{arrows.meta,positioning,calc,fit,trees,shapes.misc,backgrounds,spy,shapes.geometric}
\usepackage[ruled,vlined,linesnumbered]{algorithm2e}
\usepackage{wasysym}
\usepackage{thmtools}
\usepackage{thm-restate}
\usepackage{pgfplots}
\pgfplotsset{compat=1.18}
\usepgfplotslibrary{fillbetween}
\usepackage{hyperref}
\usepackage[nameinlink,noabbrev]{cleveref}

\theoremstyle{plain}
\newtheorem{theorem}{Theorem}
\newtheorem{lemma}[theorem]{Lemma}
\newtheorem{proposition}[theorem]{Proposition}
\newtheorem{corollary}[theorem]{Corollary}
\newtheorem{claim}[theorem]{Claim}

\theoremstyle{definition}
\newtheorem{definition}[theorem]{Definition}
\newtheorem{example}[theorem]{Example}

\theoremstyle{remark}
\newtheorem{remark}[theorem]{Remark}

\newcommand{\convexpath}[2]{
    [
    create hullcoords/.code={
        \global\edef\namelist{#1}
        \foreach [count=\counter] \nodename in \namelist {
            \global\edef\numberofnodes{\counter}
            \coordinate (hullcoord\counter) at (\nodename);
        }
        \coordinate (hullcoord0) at (hullcoord\numberofnodes);
        \pgfmathtruncatemacro\lastnumber{\numberofnodes+1}
        \coordinate (hullcoord\lastnumber) at (hullcoord1);
    },
    create hullcoords
    ]
    ($(hullcoord1)!#2!-90:(hullcoord0)$)
    \foreach [
    evaluate=\currentnode as \previousnode using int(\currentnode-1),
    evaluate=\currentnode as \nextnode using int(\currentnode+1)
    ] \currentnode in {1,...,\numberofnodes} {
        let \p1 = ($(hullcoord\currentnode) - (hullcoord\previousnode)$),
        \n1 = {atan2(\y1,\x1) + 90},
        \p2 = ($(hullcoord\nextnode) - (hullcoord\currentnode)$),
        \n2 = {atan2(\y2,\x2) + 90},
        \n{delta} = {Mod(\n2-\n1,360) - 360}
        in
        {arc [start angle=\n1, delta angle=\n{delta}, radius=#2]}
        -- ($(hullcoord\nextnode)!#2!-90:(hullcoord\currentnode)$)
    }
}

\usepackage{xcolor}
\hypersetup{
  colorlinks=true,
  citecolor=green!50!black,
  linkcolor=blue!80!black,
  urlcolor=blue!0!black
}

\newcommand{\bM}{\boldsymbol{M}}

\newcommand{\B}{\mathcal{B}}
\newcommand{\bB}{\boldsymbol{B}}
\newcommand{\bD}{\boldsymbol{D}}
\newcommand{\F}{\mathcal{F}}
\newcommand{\C}{\mathcal{C}}
\renewcommand{\L}{\mathcal{L}}
\newcommand{\I}{\mathcal{I}}
\newcommand{\X}{\mathcal{X}}
\newcommand{\Top}{\mathsf{Top}}
\newcommand{\ch}{\mathsf{ch}}
\newcommand{\rk}{\mathrm{rk}}
\newcommand{\id}{\mathrm{id}}
\newcommand{\inv}{\mathrm{inv}}

\newcommand{\dKT}{d_{\mathrm{KT}}}
\newcommand{\cost}{\mathrm{cost}}

\title{Ranking and Rank Aggregation with \\Matroid Prefix Constraints}
\author{
Seiei Ando\thanks{Department of Mathematical and Computing Science, School of Computing, Institute of Science Tokyo, Tokyo 152-8552, Japan. Email: \texttt{ando.s.373d@m.isct.ac.jp}}
\and
Yu Yokoi\thanks{Department of Mathematical and Computing Science, School of Computing, Institute of Science Tokyo, Tokyo 152-8552, Japan. Email: \texttt{yokoi@comp.isct.ac.jp}}
}
\date{}

\begin{document}
\maketitle

\begin{abstract}
We study ranking and rank aggregation under the Kendall tau distance, subject to matroid or flag matroid constraints on prefixes of the output ranking. In the matroid case, the top-$k$ prefix is required to form a base of a matroid; in the flag matroid case, several prescribed prefixes are required to form bases of a sequence of matroids linked by quotient relations. This framework contains the previously studied notions of $k$-fairness and block-fairness as special cases, and also captures more general hierarchical and assignment-type lower- and upper-quota constraints.

We provide a polynomial-time algorithm for finding, given a single input ranking, a closest feasible ranking under flag matroid prefix constraints. The algorithm is a natural greedy procedure, and its optimality is proved via a Bruhat order argument on the symmetric group. As a consequence, existing approximation frameworks for fair rank aggregation carry over to the matroidal setting. We also prove that rank aggregation with matroid constraints is NP-hard for every fixed number $m\ge 2$ of input rankings, even under partition matroid constraints.
\end{abstract}

\section{Introduction}
\label{sec:intro}
Ranking a set of alternatives and aggregating multiple rankings into a single consensus ranking are fundamental tasks in social choice, information retrieval, meta-search, recommendation, and many other applications \cite{dwor01web,wang24survey}.
Throughout the paper, a \emph{ranking} of a finite set $E$ is a total order on $E$, which we identify with a bijection $\pi:[n]\to E$ where $[n]=\{1,\dots,n\}$, $n=|E|$, and $\pi(i)$ is the element placed at position $i$.

To compare rankings, we use the {\em Kendall tau distance} (a.k.a. {\em bubble-sort distance}). For two rankings $\pi$ and $\sigma$, their Kendall tau distance, denoted by $\dKT(\pi,\sigma)$, is the number of unordered pairs of elements on which $\pi$ and $\sigma$ disagree. This is one of the most standard distances in rank aggregation and voting theory. In the unconstrained setting, minimizing the sum of Kendall tau distances to the input rankings yields the classical \emph{Kemeny rule} \cite{Kem59,YoungLevenglick78,You88}. The Kemeny aggregation rule has strong axiomatic justification, but its computational side is difficult: computing an optimal Kemeny ranking is NP-hard \cite{BartholdiTT89,hema05kemeny}. On the positive side, constant-factor approximations and even a PTAS are known \cite{ailo08agg,KenyonMathieuSchudy07}.

In this paper, we study constrained versions of ranking and rank aggregation problems. One natural source of constraints is fairness.
Recent work on fair ranking has imposed group-fair representation requirements on top positions of a ranking. For instance, Celis--Straszak--Vishnoi~\cite{CelisSV18} studied top-position fairness under metrics other than Kendall tau. In the Kendall tau setting, the \emph{closest fair ranking} (CFR) problem and the \emph{fair rank aggregation} (FRA) problem have been studied under proportionate fairness notions~\cite{chak22fair,wei22sigmod}. In the model of Chakraborty--Das--Khan--Subramanian~\cite{chak22fair}, the candidate set is partitioned into groups, and $k$-fairness requires the top-$k$ prefix to satisfy lower and upper bounds on the number of candidates from each group. The stronger notion of block-fairness imposes analogous constraints on multiple prescribed prefixes. Under these notions, CFR admits exact polynomial-time algorithms, whereas FRA has mainly been studied via approximation algorithms~\cite{chak22fair,chak25improve,wei22sigmod}.

Our starting point is the observation that the fairness constraints studied in CFR and FRA have a matroidal structure. The $k$-fairness constraint can be expressed as the requirement that the top-$k$ set forms a base of an appropriate matroid, while the stronger block-fairness constraint naturally gives rise to a flag matroid, i.e., a sequence of matroids linked by quotient relations \cite{Borovik03,CameronDMS22,JarraLorscheid24,brandenburg2024quotients}. For matroids $M$ and $N$ on the same ground set, $N$ is a {\em quotient} of $M$ if their rank functions satisfy $\rk_N(Y)-\rk_N(X)\leq \rk_M(Y)-\rk_M(X)$ for any $X\subseteq Y\subseteq E$.

\subsection*{Our Contributions}
We study matroidal generalizations of closest fair ranking (CFR) and fair rank aggregation (FRA). In \emph{Matroid-CFR}, given a ranking $\pi$ and a matroid $M$ of rank $k$ on $E$, the task is to find a ranking $\sigma$ minimizing $\dKT(\sigma, \pi)$ subject to the top $k$ elements of $\sigma$ forming a base of~$M$. In \emph{Flag-Matroid-CFR}, the constraint is given by a flag matroid $\bM=(M_1,M_2,\dots,M_s)$ with rank $k_1<k_2<\cdots<k_s$: for each $i\in[s]$, the top $k_i$ elements of the output $\sigma$ must form a base of $M_i$. The corresponding aggregation problems are \emph{Matroid-FRA} and \emph{Flag-Matroid-FRA}, where multiple rankings $\pi_1, \pi_2,\dots, \pi_m$ are given and the objective is $\sum_{i=1}^m \dKT(\sigma,\pi_i)$.

\paragraph*{Exact Closest Ranking under Flag Matroid Constraints.}
Our main contribution is an exact polynomial-time algorithm for Flag-Matroid-CFR.

Matroid-CFR, i.e., a special case with a single matroid, serves as a warm-up. Once the top-$k$ set is fixed, the optimal relative order of the chosen elements and of the remaining elements is inherited from the input ranking $\pi$, and the objective reduces to a linear weight sum over the chosen base. Thus Matroid-CFR reduces to a minimum-weight base problem (Theorem~\ref{thm:CFR}).

The main difficulty lies in the genuine flag matroid case, where several prefix constraints must be satisfied simultaneously. When expressed in terms of the prefix bases, the objective does not decompose into independent layer-wise terms, because $\dKT(\cdot,\pi)$ depends on pairwise relative positions, including pairs crossing different layers. A natural integer-programming formulation therefore leads to a quadratic objective, and we are not aware of a direct reduction to a known polynomial-time solvable combinatorial optimization problem.

Our algorithm is nonetheless a simple greedy-type algorithm. The nontrivial part is proving that its output is indeed optimal. After relabeling the candidates so that $\pi=\id$, minimizing $\dKT(\sigma,\pi)$ is equivalent to minimizing the inversion number $\inv(\sigma)$. We then compare feasible rankings via the Bruhat order on the symmetric group and prove, using the tableau criterion of Björner--Brenti~\cite{BjornerBrenti1996}, that the greedy output is the Bruhat-minimum among all feasible rankings. This establishes the correctness of our greedy-type algorithm for Flag-Matroid-CFR (Corollary~\ref{cor:flag-CFR}).

Once the closest-ranking problem becomes tractable, existing approximation frameworks for fair rank aggregation by Chakraborty--Das--Dey--Yan~\cite{chak25improve} extend to the matroidal setting because these algorithms use CFR-type procedures as subroutines. With appropriate generalizations of these procedures, we obtain a $(2+\varepsilon)$-approximation algorithm for Matroid-FRA for every constant $\varepsilon>0$ and a $2.881$-approximation algorithm for Flag-Matroid-FRA.

\paragraph*{Natural Classes of Relevant Flag Matroids.}
To demonstrate the scope of the framework, Section~\ref{sec:examples}
presents explicit classes of matroids and flag matroids that can arise naturally
in ranking and aggregation settings. While many examples of flag matroids in the
literature come from algebraic or geometric constructions
\cite{CameronDMS22,JarraLorscheid24,Borovik03} or special combinatorial
families such as lattice paths~\cite{NaturalFlag}, our examples arise from
familiar combinatorial constraint structures. Starting from a partition, a
bipartite eligibility graph, or a laminar family, we show that two compatible
parameter choices yield matroids in a quotient relation, and hence a flag
matroid. These classes model, for instance, assignment-type feasibility
constraints in reserve-based market design~\cite{DoganImamuraYenmez25} and
hierarchical group quotas in fairness-motivated selection and ranking
settings~\cite{BredereckFILS18,GorantlaMDL23}, beyond the original
partition-based constraints.

\paragraph*{Hardness of Aggregation for a Fixed Number of Voters.}
We also study matroid-constrained rank aggregation when the number $m$ of voters is fixed.
For the unconstrained Kemeny problem, the complexity with a constant number of voters has been well studied: polynomial-time solvability for $m\le 2$ is folklore, whereas NP-hardness was established for every even $m\ge 4$~\cite{dwor01web} and every odd $m\ge 7$~\cite{bach19kmajor}.
The cases $m=3$ and $m=5$ remained open \cite{bach19kmajor,krai23popular} at the time of submission, but have since been shown to be NP-hard by Madarasi~\cite{Madarasi2026} and Peters~\cite{Peters2026}.
For the constrained setting, we prove that Matroid-FRA is NP-hard for every fixed number of voters $m\ge 2$, already under partition matroid constraints.
In particular, the case with $2$ voters becomes NP-hard under these constraints, whereas its unconstrained counterpart is polynomial-time solvable.

\subsubsection*{Further Related Work}
Classical rank aggregation under Kendall tau distance has been studied extensively from the viewpoints of social choice, approximation algorithms, and parameterized complexity; see, e.g., \cite{Kem59,YoungLevenglick78,You88,ailo08agg,betz09fpt,fitz21kemeny,wang24survey,bach19kmajor,krai23popular}. On the fairness side, recent work has considered both post-processing methods for fair rankings and aggregation under group-based quota constraints \cite{ZehlikeBCM17,CelisSV18,wei22sigmod,chak22fair,chak25improve}. Our work extends this line by replacing quota systems with general matroid and flag matroid feasibility constraints.

Matroids and matroid quotients are classical objects in matroid theory \cite{Oxley11,Schrijver03}. Flag matroids are matroidal analogues of flags and form a special class of Coxeter matroids \cite{Borovik03,CameronDMS22,JarraLorscheid24}.
Recent work studies quotients in the broader framework of M-convex sets \cite{brandenburg2024quotients}. Our analysis of Flag-Matroid-CFR uses the Bruhat order on the symmetric group; for background, we refer to \cite{BjornerBrenti1996,BjornerBrenti}.

\paragraph*{Paper Organization.}
The rest of the paper is organized as follows.
We begin with preliminaries on rankings and Kendall tau distance in Section~\ref{sec:prelimi}, and on matroids and flag matroids in Section~\ref{sec:matroid_flag_matroid}.
We then study the closest ranking problems in Section~\ref{sec:matroid-CFR}.
After that, we present approximation consequences for the aggregation problems in Section~\ref{sec:approx} and prove hardness results for a fixed number of voters in Section~\ref{sec:hardness}.
Proofs of statements marked with $(\star)$ are provided in Appendices~\ref{app:quotient}--\ref{app:hardness}.

\section{Fair Ranking Problems}\label{sec:prelimi}
Let $E$ be a finite ground set with $|E|=n$, whose elements we regard as candidates to be ranked.
We use the notation $[k]=\{1,2,\dots, k\}$ for any positive integer $k$, with $[0]=\emptyset$.

A {\em ranking} of $E$ is a bijection $\pi:[n]\to E$, where $\pi(i)$ denotes the $i$th ranked element. We also identify $\pi$ with the list $(\pi(1),\pi(2),\dots,\pi(n))$ in which elements are arranged from best to worst. For $e,f\in E$, we write $e\prec_\pi f$ if $\pi^{-1}(e)<\pi^{-1}(f)$ and write $e\preceq_\pi f$ if $e\prec_\pi f$ or $e=f$. The notations $e\succ_\pi f$ and $e\succeq_\pi f$ are defined similarly.

For rankings $\pi,\sigma:[n]\to E$, the {\em Kendall tau distance} between
$\pi$ and $\sigma$ is defined by
\[
\dKT(\pi,\sigma)
=
\left|
\set{(e,f)\in E\times E|e \prec_{\pi} f\text{ and }e \succ_{\sigma} f}
\right|.
\]
Equivalently, $\dKT(\pi,\sigma)$ counts the
number of unordered pairs of elements on which $\pi$ and $\sigma$ disagree. It is well known that $\dKT$ is a metric on the set of rankings \cite{dwor01web}.

We now recall the fairness notions introduced in~\cite{chak22fair}.
These notions aim to protect minorities while preventing any group from
dominating the top positions. For a ranking $\pi:[n]\to E$ and an integer
$k\in[n]$, let $\Top_k(\pi)\coloneqq \{\pi(1),\pi(2),\dots,\pi(k)\}$ denote the set of top-$k$ candidates in $\pi$, with $\Top_0(\pi)=\emptyset$.

\begin{definition}[$(\boldsymbol{\alpha},\boldsymbol{\beta})$-$k$-fairness \cite{chak22fair}]
\label{def:weak-fair}
Let $E$ be a set of $n$ candidates partitioned into $g$ groups
$G_1,\dots,G_g$. Let
$\boldsymbol{\alpha}= (\alpha_1,\dots,\alpha_g) \in [0,1]^g$,
$\boldsymbol{\beta}= (\beta_1,\dots,\beta_g) \in [0,1]^g$, and $k \in [n]$. A ranking $\pi:[n]\to E$ is
{\em $(\boldsymbol{\alpha},\boldsymbol{\beta})$-$k$-fair (or $k$-fair)} if
\[
\lfloor \alpha_i k \rfloor \le |\Top_k(\pi)\cap G_i|
\le \lceil \beta_i k \rceil
\qquad \text{for all } i\in[g].
\]
\end{definition}
This notion imposes the constraint only at a single prefix length, namely $k$. The next notion strengthens this requirement by imposing analogous constraints on multiple prefixes.

\begin{definition}[$(\boldsymbol{\alpha},\boldsymbol{\beta})$-block-$k$-fairness \cite{chak22fair}]
\label{def:block-fair}
Let $E$ be a set of $n$ candidates partitioned into $g$ groups
$G_1,\dots,G_g$. Let
$\boldsymbol{\alpha}= (\alpha_1,\dots,\alpha_g) \in [0,1]^g$,
$\boldsymbol{\beta}= (\beta_1,\dots,\beta_g) \in [0,1]^g$, and $k,b\in[n]$, where we assume that $b\alpha_i$ and $b\beta_i$ are integers for all
$i\in[g]$. A ranking $\pi:[n]\to E$ is
{\em $(\boldsymbol{\alpha},\boldsymbol{\beta})$-block-$k$-fair (or block-$k$-fair)} if
\[
\alpha_i h \le |\Top_h(\pi)\cap G_i| \le \beta_i h
\quad
\text{for all } i\in[g] \text{ and all } h\in[n]
\text{ with } h\ge k \text{ and } h\equiv 0 \!\!\pmod{b}.
\]
\end{definition}

The {\em Closest Fair Ranking (CFR)} problem, investigated in~\cite{chak22fair}, asks, given a ranking $\pi$ on $E=\bigsqcup_{i=1}^g G_i$, vectors $\boldsymbol{\alpha},\boldsymbol{\beta}\in [0,1]^g$, and an integer $k$, to find a $k$-fair ranking $\sigma$ minimizing the Kendall tau distance $\dKT(\sigma,\pi)$.

The {\em Fair Rank Aggregation (FRA)} problem, investigated in~\cite{chak22fair,chak25improve}, asks, given $m$ rankings $\pi_1, \dots, \pi_m$ on $E=\bigsqcup_{i=1}^g G_i$, vectors $\boldsymbol{\alpha},\boldsymbol{\beta}\in [0,1]^g$, and an integer $k$, to find a $k$-fair ranking $\sigma$ minimizing $\sum_{j=1}^m \dKT(\sigma,\pi_j)$.

The block-fairness versions of these problems can also be defined.
Efficient algorithms for CFR were given for the $k$-fair and block-$k$-fair versions \cite{chak22fair}. In contrast, FRA is a generalization of the NP-hard Kemeny consensus problem, and previous work has focused on approximation algorithms \cite{chak22fair,chak25improve}.

\section{Matroids and Flag Matroids}\label{sec:matroid_flag_matroid}
We briefly recall standard facts on matroids and then introduce some basics of flag matroids; see, e.g.,
\cite{Oxley11,Schrijver03} for more details on matroids and \cite{CameronDMS22,JarraLorscheid24,Borovik03} for details on flag matroids.
We also provide examples of flag matroids that arise from combinatorial structures.
\subsection{Matroids}

A {\em matroid} is a pair $M=(E,\I)$ consisting of a finite set $E$ and a family
$\I\subseteq 2^E$ satisfying the following three axioms:
(I1) $\emptyset\in \I$,
(I2) If $I\subseteq J\in \I$, then $I\in \I$, and
(I3) If $I,J\in \I$ and $|I|<|J|$, then there exists
$e\in J\setminus I$ such that $I\cup\{e\}\in \I$.

The set $E$ is called the {\em ground set}, and $\I$ the {\em independent set
family} of $M$. The {\em base family} $\B$ of $M$ consists of the
inclusion-wise maximal members of $\I$. By (I3), all members of $\B$ have the
same cardinality, say $k$; this number is called the {\em rank} of $M$.

A matroid can be specified by its base family, i.e., when $\B$ is the base family of a matroid, the independent set family is given as $\I=\set{I\subseteq E |\exists B\in \B: I\subseteq B}$.

In this paper, we mainly use the base-family viewpoint; we suppose that a matroid is represented in the form $M=(E, \B)$ using the base family. For algorithmic purposes, however, we assume access to an independence oracle, following the usual convention.

The rank function $\rk_M:2^E\to \mathbb{Z}_+$ of a matroid $M=(E, \B)$ is given by
\[\rk_M(S)=\max\{\,|S\cap B|\mid B\in \B\,\}\quad (S\subseteq E).\]
It is also written with the independent set family $\I$ as $\rk_M(S)=\max\{\,|I|\mid I\in \I,\, I\subseteq S\,\}$.
\medskip

Below, we fix a finite set $E$ and write $n=|E|$.
For any ranking $\pi:[n]\to E$ and any integer $k\in[n]$, the
{\em Gale order on $\binom{E}{k}$ induced by $\pi$} is the partial order
$\preceq_\pi^{\mathrm G}$ on $\binom{E}{k}$ defined as follows.
For $A,B\in \binom{E}{k}$, write
\begin{align*}
A&=\{a_1,a_2,\dots,a_k\},
\quad a_1\prec_\pi a_2\prec_\pi \cdots \prec_\pi a_k,\\
B&=\{b_1,b_2,\dots,b_k\},
\quad b_1\prec_\pi b_2\prec_\pi \cdots \prec_\pi b_k,
\end{align*}
and define $\preceq_\pi^{\mathrm G}$ by
\[
A\preceq_\pi^{\mathrm G} B
\quad\Longleftrightarrow\quad a_i\preceq_\pi b_i\ \ (i=1,2,\dots, k).
\]
One reformulation of Gale's theorem \cite{Gale68} states that
a family $\mathcal{B}\subseteq \binom{E}{k}$ is the base family of a matroid on
$E$ if and only if, for every ranking $\pi$ over $E$, the family
$\mathcal{B}$ has a unique $\preceq_\pi^{\mathrm G}$-minimum member.

The unique $\preceq_\pi^{\mathrm G}$-minimum base is known to be computable by the {\em greedy algorithm}. For a matroid $M=(E,\B)$ and a ranking $\pi$, the greedy algorithm on $M$ with respect to $\pi$
works as follows: start with $X=\emptyset$, and for each $i=1,2,\dots,n$, if
$X\cup\{\pi(i)\}\in \I$, then update $X\gets X\cup\{\pi(i)\}$. The output is the final set $X$. Clearly, this algorithm requires $O(n)$ oracle calls.
The algorithm is also known to compute a minimum-weight base. We summarize these well-known facts in the following proposition; see \cite[Theorem~40.1]{Schrijver03} and also \cite{Gale68,Vince02}.

\begin{proposition}\label{prop:basic}
Let $M=(E, \B)$ be a matroid and $\pi:[n]\to E$ be a ranking.
The output of the greedy algorithm on $M$ with respect to $\pi$ coincides with
the unique $\preceq_\pi^{\mathrm G}$-minimum base $B^*\in \B$.
In addition, $B^*$ is a minimum-weight base with respect to any weight function compatible with $\pi$ (i.e., $w\in\mathbb{R}^E$ such that $e\prec_{\pi} f\implies w(e)\leq w(f)$).
\end{proposition}

\subsection{Flag Matroids}\label{sec:flag}
As in the previous subsection, $E$ is a finite set with $|E|=n$.
Let $\boldsymbol{k}=(k_1,\dots,k_s)$ be a strictly increasing sequence of integers with
$0\le k_1<\cdots<k_s\le n$.
A {\em flag of rank $\boldsymbol{k}$ on $E$} is a chain
\[
F_1\subsetneq F_2\subsetneq \cdots \subsetneq F_s \subseteq E
\qquad\text{with}\qquad |F_i|=k_i \ \ (i=1,2,\dots,s).
\]
We write $\F^{\boldsymbol{k}}_{E}$ for the set of all such flags.

A {\em flag matroid} of rank $\boldsymbol{k}$ on $E$ is a sequence $\bM=(M_1,\dots,M_s)$ of matroids on the same ground set $E$ such that
$M_i$ has rank $k_i$ for each $i\in[s]$ and $M_i$ is a quotient of $M_{i+1}$ for each $i\in[s-1]$.
For matroids $M$ and $N$ on the same ground set, $N$ is a
{\em quotient} of $M$ if
\[\rk_N(Y)-\rk_N(X)\leq \rk_M(Y)-\rk_M(X)\quad \text{for all $X\subseteq Y\subseteq E$}.\]
Various equivalent definitions of matroid quotient in terms of bases, circuits, and flats are known; see, e.g., \cite[Definition~1.1]{brandenburg2024quotients}, \cite[Proposition~7.4.7]{Bry86}, and \cite[Definition~44]{CameronDMS22}. If $\B_i$ denotes the base family of $M_i$, then the associated family of {\em flag bases} is
\[
\mathcal{F}(\bM)
=
\set{(B_1,B_2,\dots,B_s)| B_i\in \B_i \ (i=1,2,\dots,s),\ B_1\subsetneq B_2\subsetneq \cdots \subsetneq B_s}.
\]

As with ordinary matroids, flag matroids admit a characterization via Gale order. To state it, we extend the Gale order from $\binom{E}{k}$ to $\F^{\boldsymbol{k}}_{E}$. For a ranking $\pi$ over $E$ and any two flags
$\boldsymbol{F}=(F_1,\dots,F_s), \boldsymbol{G}=(G_1,\dots,G_s)\in \F^{\boldsymbol{k}}_{E}$, we define $\preceq_\pi^{\mathrm G}$ by
\[
\boldsymbol{F}\preceq_{\pi}^{\mathrm G} \boldsymbol{G}
\quad\Longleftrightarrow\quad
F_h\preceq_{\pi}^{\mathrm G} G_h
\ \ (h=1,2,\dots,s).
\]

Similarly to the case of matroids, a flag matroid is characterized by the existence of a minimum flag base with respect to Gale order \cite[Section~3.3 and Theorem~3.2]{BorovikGelfandStone02};
a nonempty family $\mathcal{F}\subseteq \F^{\boldsymbol{k}}_{E}$ is the
family of flag bases of a flag matroid on $E$ if and only if, for every ranking
$\pi$ over $E$, the family $\mathcal{F}$ has a unique
$\preceq_{\pi}^{\mathrm G}$-minimum member.

For algorithmic purposes, we provide here a constructive proof of the ``only if'' direction.

\begin{proposition}\label{prop:optimalflag}
For a flag matroid $\bM=(M_1,\dots,M_s)$ on $E$ and a ranking $\pi:[n]\to E$, the outputs $B_i^*~(i\in[s])$ of the greedy algorithm on $M_i$ with respect to $\pi$ form a chain $B_1^*\subsetneq B_2^*\subsetneq \cdots \subsetneq B_s^*$, and $(B_1^*, B_2^*, \cdots, B_s^*)$ is the unique $\preceq_{\pi}^{\mathrm G}$-minimum flag base.
\end{proposition}
\begin{proof}
For each $i\in[s]$, let $\I_i$ denote the independent set family of $M_i$. Consider running the greedy algorithm on all $M_i$ simultaneously. Start with $(X_1,\dots, X_s)=(\emptyset,\dots, \emptyset)$, and scan $\pi(j)$ for $j=1,2,\dots,n$ sequentially as follows: for each $i\in [s]$, if $X_i\cup\{\pi(j)\}\in \I_i$, then update $X_i\gets X_i\cup\{\pi(j)\}$. After scanning the last element $\pi(n)$, output the final tuple $(X_1,\dots, X_s)$. Clearly, each $X_i$ in the output coincides with $B_i^*$, since its construction is exactly the usual greedy algorithm on $M_i$. We now show by induction that $(X_1,\dots, X_s)$ is always a chain during the algorithm.

The claim clearly holds before the algorithm scans $\pi(1)$. For any $j\in[n]$, suppose that the claim holds just before the algorithm scans $\pi(j)$. Observe that, just before $\pi(j)$ is scanned, for any $i\in [s]$, the set $X_i$ is an inclusion-wise maximal independent subset of $\Top_{j-1}(\pi)=\{\pi(1),\dots, \pi(j-1)\}$ in $M_i$ (this follows from (I2)). Therefore, $X_i\cup\{\pi(j)\}\in \I_i$ holds if and only if $\rk_{M_i}(\Top_{j}(\pi))-\rk_{M_i}(\Top_{j-1}(\pi))=1$ (this follows from (I3) and the definition of $\rk_{M_i}$). By repeated application of the quotient property, for any $i'>i$, if $\rk_{M_i}(\Top_{j}(\pi))-\rk_{M_i}(\Top_{j-1}(\pi))=1$, then $\rk_{M_{i'}}(\Top_{j}(\pi))-\rk_{M_{i'}}(\Top_{j-1}(\pi))=1$, which is equivalent to $X_{i'}\cup\{\pi(j)\}\in \I_{i'}$. Therefore, whenever $\pi(j)$ is added to $X_i$, it is also added to $X_{i'}$ for all $i'>i$.
Thus, $(X_1,\dots, X_s)$ remains a chain.

By Proposition~\ref{prop:basic}, for each $i\in[s]$, the set $B_i^*$ is the unique $\preceq_{\pi}^{\mathrm G}$-minimum base of $M_i$. Hence, for any flag base $(B_1,\dots,B_s)\in \mathcal{F}(\bM)$, we have $B_i^*\preceq_{\pi}^{\mathrm G} B_i$ for all $i\in[s]$, and hence $(B_1^*,\dots,B_s^*)$ is the unique $\preceq_{\pi}^{\mathrm G}$-minimum flag base.
\end{proof}

\subsection{Examples of Matroids and Flag Matroids}\label{sec:examples}
We present three classes of matroids relevant to constrained ranking and aggregation, and identify sufficient conditions under which members of the same class are in a quotient relation.

The first example directly captures $k$-fairness; see \cite[Theorem~12]{HorschSzigeti21} and \cite[p.~182]{Frankbook}.
\begin{example}[Generalized Partition Matroid]\label{ex:g-partition}
Let $E=G_1\sqcup \cdots \sqcup G_g$ be a partition, let $\boldsymbol{\ell}=(\ell_i)_{i\in [g]}\in \mathbb{Z}_+^g$ and
$\boldsymbol{u}=(u_i)_{i\in [g]}\in \mathbb{Z}_+^g$ with $0\le \ell_i \le u_i \le |G_i|~(i\in[g])$, and let $k\in \mathbb{Z}_{+}$ satisfy $\sum_{i=1}^g \ell_i \le k \le \sum_{i=1}^g u_i$.
Then the family
\[
\B_{k,\boldsymbol{\ell},\boldsymbol{u}}
=
\set{B\subseteq E |~ |B|=k,\ \ \ell_i \le |B\cap G_i|\le u_i
\ \text{for all } i\in[g]\,}
\]
is the base family of a matroid, called a {\em generalized partition matroid}.
\end{example}

A usual partition matroid is the case $\ell_i=0~(i\in[g])$ and $k=\sum_{i\in[g]} u_i$.

The $k$-fairness condition in Section~\ref{sec:prelimi} is captured by a generalized partition matroid: if we set $\ell^*_i=\lfloor \alpha_i k \rfloor$ and $u^*_i=\lceil \beta_i k \rceil$ for each $i\in[g]$, then a ranking $\pi$ over $E$ is $(\boldsymbol{\alpha},\boldsymbol{\beta})$-$k$-fair if and only if $\Top_k(\pi)\in \B_{k,\boldsymbol{\ell}^*, \boldsymbol{u}^*}$, provided that this family is nonempty.

The next proposition gives a sufficient condition for quotient relations between generalized partition matroids, and hence for obtaining flag matroids. The proof is in Appendix~\ref{app:quotient}.
\begin{restatable}[$\star$]{proposition}{PropPartitionQuotient}\label{prop:quotient-partition}
Let $N$ and $M$ be generalized partition matroids defined from the same partition $E=G_1\sqcup \cdots \sqcup G_g$, but with different parameters $(k, \boldsymbol{\ell}, \boldsymbol{u})$ and $(k', \boldsymbol{\ell}', \boldsymbol{u}')$, respectively, where $0<k<k'$ and neither $\B_{k,\boldsymbol{\ell},\boldsymbol{u}}$ nor $\B_{k',\boldsymbol{\ell}',\boldsymbol{u}'}$ is empty.
Then $N$ is a quotient of $M$ if
\[\ell_i\leq \ell'_i\leq \lfloor\tfrac{k'}{k}\ell_i\rfloor,
\qquad
u_i\leq u'_i\leq \lfloor\tfrac{k'}{k}u_i \rfloor
\qquad
\text{for all } i\in[g].
\]
\end{restatable}

Consequently, any sequence $M_1,\dots,M_s$ such that each pair $(M_j,M_{j+1})$ satisfies the sufficient condition in Proposition~\ref{prop:quotient-partition} forms a flag matroid. In particular, block-$k$-fairness is a special case; if we define a flag matroid $\bM=(M_1, \dots, M_s)$ and its rank $(h_1, \dots, h_s)$ appropriately, we have that a ranking $\pi$ is $(\boldsymbol{\alpha},\boldsymbol{\beta})$-block-$k$-fair if and only if $(\Top_{h_1}(\pi), \dots, \Top_{h_s}(\pi))\in \F(\bM)$. (See Corollary~\ref{cor:block-flag} in Appendix~\ref{app:quotient} for the formal statement and proof.)

\medskip
We next give two more general examples.

One is a generalized transversal matroid, which is the size-$k$ truncation of the generalized matroid in \cite[Example~2]{yoko17gpoly}. It models assignment-type feasibility constraints. For example, the top-$k$ prefix of a consensus ranking may have to be assignable to eligible departments, tracks, or reserved positions under capacity or reserve
requirements. A related reserve-graph formulation uses transversal matroids to encode eligibility for reserved positions~\cite{DoganImamuraYenmez25}; our generalized version further allows lower and upper quotas for each position type.
\begin{example}[Generalized Transversal Matroid]\label{example-transversal}
Let $E$ and $D$ be disjoint finite sets, let $\Gamma:D\to 2^E$, let
$\boldsymbol{p}=(p_d)_{d\in D}\in \mathbb{Z}_+^D$ and
$\boldsymbol{q}=(q_d)_{d\in D}\in \mathbb{Z}_+^D$ with $0\leq p_d\leq q_d\leq |\Gamma(d)|~(d\in D)$, and let $k\in \mathbb{Z}_+$. Define
\[
\B^{\rm tv}_{k,\boldsymbol{p},\boldsymbol{q}}
=
\left\{\,
B\subseteq E
\ \middle|\
\begin{array}{l}
|B|=k,\ \exists \varphi:B\to D \text{ such that}\\
\qquad \varphi^{-1}(d)\subseteq \Gamma(d),\
p_d\le |\varphi^{-1}(d)|\le q_d
~~\,\text{for all}~d\in D
\end{array}
\right\}.
\]
If nonempty, $\B^{\rm tv}_{k,\boldsymbol{p},\boldsymbol{q}}$ is the base family of a matroid, called a {\em generalized transversal matroid}.
\end{example}
When $p_d=0$ and $q_d=1$ for all $d\in D$, this reduces to the size-$k$ truncation of the usual {\em transversal matroid}.

We give a sufficient condition for two generalized transversal matroids defined from the same bipartite structure to have a quotient relation.
\begin{restatable}[$\star$]{proposition}{PropTransversalQuotient}\label{prop:quotient-transversal}
Let $N$ and $M$ be generalized transversal matroids on $E$ defined with the
same $D$ and $\Gamma$, but with different parameters
$(k,\boldsymbol{p},\boldsymbol{q})$ and
$(k',\boldsymbol{p}',\boldsymbol{q}')$, respectively, where $0<k<k'$ and neither
$\B^{\rm tv}_{k,\boldsymbol{p},\boldsymbol{q}}$ nor
$\B^{\rm tv}_{k',\boldsymbol{p}',\boldsymbol{q}'}$ is empty.
Then $N$ is a quotient of $M$ if
\[
p_d\le p'_d\le \lfloor\tfrac{k'}{k}\,p_d\rfloor,
\qquad
q_d\le q'_d\le \lfloor\tfrac{k'}{k}\,q_d\rfloor
\qquad
\text{for all } d\in D.
\]
\end{restatable}

We provide the proof of Proposition~\ref{prop:quotient-transversal} in Appendix~\ref{app:quotient} and describe only its main ingredient here. The proof is based on induction by a bipartite graph, an operation that constructs a new matroid from a given one. Brandenburg--Loho--Smith~\cite{brandenburg2024quotients} showed that induction preserves quotient relations. Proposition~\ref{prop:quotient-transversal} then follows by applying Proposition~\ref{prop:quotient-partition} to suitable generalized partition matroids and then using induction.

A second example is what we call a generalized laminar matroid. It is the
fixed-cardinality restriction of the laminar lower/upper-quota family, which is
a generalized matroid~\cite[Example~1]{yoko17gpoly}. Unlike the
partition-based constraints, these constraints allow nested groups, so one
can impose quotas simultaneously on broad groups and their subgroups. Similar
hierarchical constraints have appeared in fairness-motivated committee selection
and ranking. Bredereck et al.~\cite{BredereckFILS18} study committee selection
with label-wise cardinality constraints and explicitly discuss laminar label
structures such as countries and continents; Gorantla et al.~\cite{GorantlaMDL23}
consider laminar protected groups in fair ranking.

\begin{example}[Generalized Laminar Matroid]\label{example-laminar}
Let $E$ be a finite set, $\L\subseteq 2^E$ be a laminar family,
$g,f:\L\to \mathbb{Z}_{+}$ be functions with $g(L)\leq f(L)$ for all $L\in \L$,
and $k\in \mathbb{Z}_+$. Define
\[
\B^{\rm lam}_{k,g,f}
=\set{ B\subseteq E| ~|B|=k,~~g(L)\leq |B\cap L|\leq f(L) \text{~for all~} L\in \L}.
\]
If nonempty, $\B^{\rm lam}_{k,g,f}$ forms the base family of a matroid, called a {\em generalized laminar matroid}.
\end{example}
When $g(L)=0$ for all $L\in \L$, this reduces to a size-$k$ truncation of the usual {\em laminar matroid}.

We give a sufficient condition for two generalized laminar matroids defined on the same laminar family to have a quotient relation.
\begin{restatable}[$\star$]{proposition}{PropLaminarQuotient}\label{prop:quotient-laminar}
Suppose that generalized laminar matroids $N$ and $M$ are defined with the
same laminar family $\L$, but with different parameters
$(k,g,f)$ and
$(k',g',f')$, respectively, where $0<k<k'$ and neither
$\B^{\rm lam}_{k,g,f}$ nor
$\B^{\rm lam}_{k',g',f'}$ is empty.
Then $N$ is a quotient of $M$ if
\[
g'(L)=\lfloor \tfrac{k'}{k}\,g(L)\rfloor,
\qquad
f'(L)=\lfloor \tfrac{k'}{k}\,f(L)\rfloor
\qquad
\text{for all } L\in \L.
\]
\end{restatable}

We remark that, in contrast to the partition case, the condition in Proposition~\ref{prop:quotient-laminar} cannot in general be weakened to inequalities analogous to those in Proposition~\ref{prop:quotient-partition}. To see that neither equality can be relaxed in this way, let $E=A\sqcup\{b\}\sqcup C$, where $|A|=|C|=4$, and set $B=A\cup\{b\}$ and $\L=\{A,B\}$. Define a matroid $N$ of rank $k=3$ by setting $((g(A),f(A)), (g(B), f(B)))=((1,2),(2,3))$. For $k'=6$, the matroid $M$ in Proposition~\ref{prop:quotient-laminar} is defined by $((g'(A),f'(A)), (g'(B), f'(B)))=((2,4),(4,6))$. If we decrease the upper quota $f'(B)$ from $6$ to $4$ (resp., the lower quota $g'(B)$ from $4$ to $2$), then the corresponding inequalities are satisfied, but the resulting rank-$k'$ matroid $M'$ satisfies $\rk_N(B)-\rk_N(A)=1>0=\rk_{M'}(B)-\rk_{M'}(A)$ (resp., $\rk_N(\{b\}\cup C)-\rk_N(C)=1>0=\rk_{M'}(\{b\}\cup C)-\rk_{M'}(C)$). Thus $N$ is not a quotient of $M'$.

\begin{remark}
For all examples of matroids provided in this subsection, the independence oracle can be implemented from the defining structure and parameters using network flow algorithms. In particular, given any ranking $\pi$, the $\preceq_\pi^{\mathrm G}$-minimum base can be computed efficiently; see, e.g., \cite[Appendix~C]{yoko17gpoly}.
\end{remark}

Other sources of flag matroids are also available. For partition matroids, i.e., partition constraints without lower quotas, the weaker condition $u_i\leq u'_i~(i\in [g])$ suffices. Similar relaxations are possible in the transversal and laminar settings when lower quotas are absent. Another source comes from matroid union: if $N_1, N_2, \dots, N_s$ are matroids on the same ground set and $M_i$ is the matroid union of $N_1, \dots, N_i$ for each $i=1,\dots,s$, then $(M_1, M_2,\dots, M_s)$ forms a flag matroid (after removing repetitions of the same matroid if any).

\section{Closest Ranking with Flag Matroid Constraints}\label{sec:matroid-CFR}

We propose the following two problems as generalizations of Closest Fair Ranking (CFR) studied in \cite{chak22fair,chak25improve} and reviewed in Section~\ref{sec:prelimi}.
The first generalizes CFR with $k$-fairness.
\begin{tcolorbox}[colback=white,colframe=black,boxrule=0.2pt,arc=0pt]
\textsc{Matroid Closest Feasible Ranking (Matroid-CFR)}\\[2pt]
\textbf{Input:} A ranking $\pi$ over $E$ and a matroid $M=(E,\B)$ of rank $k$.\\
\textbf{Output:} A ranking $\sigma$ over $E$ minimizing $d_{\mathrm{KT}}(\sigma,\pi)$ subject to $\Top_k(\sigma)\in \B$.
\end{tcolorbox}
The second problem below generalizes CFR with block-$k$-fairness.

\begin{tcolorbox}[colback=white,colframe=black,boxrule=0.2pt,arc=0pt]
\textsc{Flag Matroid Closest Feasible Ranking (Flag-Matroid-CFR)}\\[2pt]
\textbf{Input:} A ranking $\pi$ over $E$ and a flag matroid $\bM=(M_1, \dots, M_s)$ of rank $\boldsymbol{k}=(k_1, \dots, k_s)$, where the base family of $M_i$ is denoted by $\B_i$.\\
\textbf{Output:} A ranking $\sigma$ over $E$ minimizing $d_{\mathrm{KT}}(\sigma,\pi)$ subject to $\Top_{k_i}(\sigma)\in \B_i~(i\in[s])$.
\end{tcolorbox}
Note that Matroid-CFR is a special case of Flag-Matroid-CFR with $s=1$.
However, we will treat this special case separately, since it is much easier than the case $s\geq 2$.

We first provide a basic observation that can be used in both problems.
Given an instance of Flag-Matroid-CFR as above, any ranking $\sigma$ satisfying $\Top_{k_i}(\sigma)\in \B_i~(i\in[s])$ is determined by the following two pieces of data: (1) a choice of a flag base $(B_1, B_2,\dots, B_s)\in \mathcal{F}(\bM)$ such that $B_i=\Top_{k_i}(\sigma)~(i\in[s])$, and (2) an ordering of each of the parts $B_i \setminus B_{i-1}$ for $i\in[s+1]$, where for convenience we set $B_0=\emptyset$ and $B_{s+1}=E$.

For any flag base $\bB=(B_1, B_2, \dots, B_s)$ of $\bM$, which forms a chain by definition, we define a ranking $\pi_{\bB}$ on $E$ as follows:
the elements are ordered by concatenating $B_1$, $B_2 \setminus B_1$, $B_3 \setminus B_2$, $\dots$, $B_s \setminus B_{s-1}$, and $E \setminus B_s$;
within each part, the elements are ordered according to the input ranking $\pi$. See Figure~\ref{fig:block}.

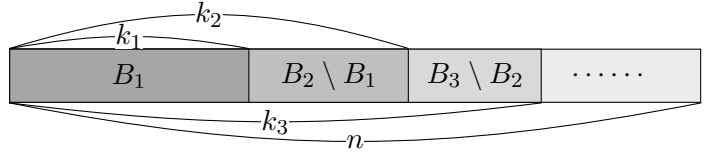
\begin{figure}[h]
\centering
\begin{tikzpicture}[every node/.style={draw, inner sep=0pt, node distance = 5pt, text centered}]
\node [rectangle, minimum height=20pt, minimum width = 260pt, xshift = 5pt, fill=gray!15](A){};
\draw (node cs:name=A,anchor=south west) to[out=-11, in=-169, edge node={node[draw = none, fill = white]{\(n\)}}] (node cs:name=A,anchor=south east);

\node [rectangle, minimum height=20pt, minimum width = 200pt, xshift = -25pt, fill=gray!30](B){};
\draw (node cs:name=B,anchor=south west) to[out=-7, in=-173, edge node={node[draw = none, fill = white]{\(k_3\)}}] (node cs:name=B,anchor=south east);

\node [rectangle, minimum height=20pt, minimum width = 150pt, xshift = -50pt, fill=gray!50](C){};
\draw (node cs:name=C,anchor=north west) to[out=17, in=163, edge node={node[draw = none, fill = white]{\(k_2\)}}] (node cs:name=C,anchor=north east);
\node [rectangle, draw = none, right =11pt of B]{\(\cdots \cdots\)};

\node [rectangle, minimum height=20pt, minimum width = 90pt, xshift = -80pt, fill=gray!70](D){};
\draw (node cs:name=D,anchor=north west) to[out=10, in=170, edge node={node[draw = none, fill = white]{\(k_1\)}}] (node cs:name=D,anchor=north east);

\node[draw=none] at (D.center) {\(B_1\)};
\node[draw=none] at ($(D.east)!0.5!(C.east)$) {\(B_2 \setminus B_1\)};
\node[draw=none] at ($(C.east)!0.5!(B.east)$) {\(B_3 \setminus B_2\)};
\end{tikzpicture}
\caption{A construction of $\pi_{\bB}$ for a ranking $\pi$ and a flag base $\bB=(B_1,B_2,\dots, B_s)$. In each block, elements are arranged according to $\pi$.}
\label{fig:block}
\end{figure}

\begin{restatable}[$\star$]{observation}{ObsOptRanking}\label{obs:fixedbase}
For a ranking $\pi:[n]\to E$ and any flag base $\bB=(B_1, B_2,\dots, B_s)$, among rankings $\sigma$ satisfying $\Top_{k_i}(\sigma)=B_i~(i\in[s])$, the one minimizing $d_{\mathrm{KT}}(\sigma, \pi)$ is $\pi_{\bB}$.
\end{restatable}

The proof is given in Appendix~\ref{app:CFR}, but it is summarized as follows.
The constraints $\Top_{k_i}(\sigma)=B_i~(i\in[s])$ fix the order of the blocks
$B_1$, $B_2 \setminus B_1$, $B_3 \setminus B_2$, $\dots$, $B_s \setminus B_{s-1}$, and $E \setminus B_s$ in $\sigma$.
Hence, among the pairs counted in $d_{\mathrm{KT}}(\sigma,\pi)$, those from different blocks contribute a constant independent of the internal ordering of each block. Therefore, it suffices to minimize the contribution of pairs within each block, which is achieved by ordering the elements of every block according to $\pi$.

Therefore, once a flag base $\bB=(B_1,\dots,B_s)\in \mathcal{F}(\bM)$ is fixed, the ranking $\sigma$ minimizing $d_{\mathrm{KT}}(\sigma,\pi)$ is uniquely determined as $\pi_{\bB}$. Hence, Flag-Matroid-CFR reduces to the problem of finding a flag base $\bB\in \mathcal{F}(\bM)$ that minimizes $d_{\mathrm{KT}}(\pi_{\bB},\pi)$.

\subsection*{Single Matroid Case}
We first apply Observation~\ref{obs:fixedbase} to the case $s=1$, i.e., Matroid-CFR.
For this special case, the objective function can be written explicitly.

\begin{restatable}[$\star$]{lemma}{LemOptRanking}\label{lem:MatroidCFR}
Given an instance $(\pi, (E, \B))$ of Matroid-CFR, where the rank of $(E, \B)$ is $k$, for any base $B\in \B$, we have $\dKT(\pi_{B}, \pi)= \sum_{e \in B} \pi^{-1}(e)-\frac{k(k+1)}{2}$.
\end{restatable}
The proof is given in Appendix~\ref{app:CFR}, but the idea is simple.
Starting from $\pi$, consider obtaining $\pi_B$ by moving the elements of $B$ to the top $k$
positions while preserving their relative order in $\pi$.
If $e\in B$ is the $h$th highest element of $B$ in $\pi$, then it moves left by
exactly $\pi^{-1}(e)-h$ positions, creating one inversion at each step. Summing over all $e\in B$ gives the formula.

Since $-\frac{k(k+1)}{2}$ is a constant independent of the choice of $B$, Lemma~\ref{lem:MatroidCFR} shows that Matroid-CFR reduces to the minimum-weight base problem with weights $w(e)=\pi^{-1}(e)$. Then Proposition~\ref{prop:basic} immediately implies the following.

\begin{theorem}\label{thm:CFR}
Given an instance $(\pi, (E, \B))$ of Matroid-CFR, let $B^*$ denote the unique $\preceq_\pi^{\mathrm G}$-minimum base in $\B$. Then $\pi_{B^*}$ is the unique optimal solution to Matroid-CFR. Thus, Matroid-CFR can be solved efficiently.
\end{theorem}

\subsection*{Difficulty with Multiple Matroids}
By Proposition~\ref{prop:optimalflag}, for a flag matroid $\bM=(M_1, M_2,\dots, M_s)$, the unique $\preceq_\pi^{\mathrm G}$-minimum bases $B_{i}^*$ of the matroids $M_i$ form a chain $B_1^*\subsetneq B_2^*\subsetneq \cdots \subsetneq B_s^*$, and hence give a flag base $\bB^*=(B^*_1,\dots, B^*_s)$. As a generalization of Theorem~\ref{thm:CFR}, we expect that $\pi_{\bB^*}$ gives the unique optimal solution to Flag-Matroid-CFR. This is indeed true, as we will show. However, the single-matroid proof does not extend directly.

We demonstrate that analyzing Flag-Matroid-CFR is not straightforward when $s\geq 2$. Take any flag base $\bB=(B_1, B_2, \dots, B_s)$ and consider the corresponding ranking $\pi_{\bB}$, which arranges the blocks $B_1, B_2\setminus B_1, \dots, B_s\setminus B_{s-1}$ and $E\setminus B_s$ in this order, where within each block the elements are arranged in the order of $\pi$. Let us write $\dKT(\pi_{\bB}, \pi)$ in a way mimicking Lemma~\ref{lem:MatroidCFR}. One possible expression is then
\[
\dKT(\pi_{\bB}, \pi)=\sum_{e\in B_1}\pi^{-1}(e)~+\!\!\!\sum_{e\in B_2\setminus B_1}(\pi|_{E\setminus B_1})^{-1}(e)~+~\cdots ~+\!\!\!\sum_{e\in B_s\setminus B_{s-1}}(\pi|_{E\setminus B_{s-1}})^{-1}(e)~+ \text{constant}
\]
where the notation $(\pi|_X)^{-1}(e)$ represents the rank of $e$ among $X\subseteq E$. Here, the first term is simply the weight of $B_1$ with respect to the weights $w(e)=\pi^{-1}(e)$, but the weights for the $i$th term $(i\geq 2)$ are affected by the choice of $B_{i-1}$. It is then possible that taking a non-optimal earlier base reduces the latter terms; for example, consider an instance where $s=2$, $\pi$ is the identity on $E=\{1,2,3\}$, and $\B_1=\{\{1\}, \{3\}\}$ and $\B_2=\{\{1,3\}\}$. Then $B_2$ must be $\{1,3\}$. Choosing $B_1=\{1\}$ minimizes the first term but then the second term is $2$ as $(\pi|_{\{2,3\}})^{-1}(3)=2$, whereas choosing $B_1=\{3\}$ makes the second term smaller (i.e., $1$) since $(\pi|_{\{1,2\}})^{-1}(1)=1$. Thus, we cannot divide the objective value into $s$ separate linear optimization problems, and the optimality of each $B_i$ with respect to weights $\pi^{-1}$ does not simply imply the optimality of the whole objective value.

Indeed, if we represent the objective function using binary variables $x_{ei}$ such that $x_{ei}=1\Leftrightarrow e\in B_i\setminus B_{i-1}$ for each  $e\in E$ and $i\in[s+1]$ (with $B_0=\emptyset$ and $B_{s+1}=E$), then this natural formulation yields a quadratic objective function. For a pair of elements $e,f$ with $e\prec_\pi f$, its contribution to the objective function can be written as $\sum_{i,j\in [s+1]: i<j} x_{fi}x_{ej}$. Summing over all pairs $e\prec_\pi f$, we obtain a quadratic objective function.

\smallskip
Another possible strategy for proving the optimality of $\pi_{\bB^*}$ is to suppose, for contradiction, that an optimal solution $\pi_{\bB'}$ differs from $\pi_{\bB^*}$, and then use the matroid exchange property to modify $\pi_{\bB'}$ so as to obtain a solution with a strictly better objective value.\footnote{
A related exchange-based argument appears in the correctness proof of the algorithm for CFR with block-fairness~\cite[Claims~3.8 and~3.9]{chak22fair}. To the best of our understanding, some additional justification is needed to ensure that the proposed swap preserves all block-prefix constraints; see Appendix~\ref{app:exchange} for a detailed discussion.}
However, the usual matroid exchange argument does not seem to extend directly, because to modify a flag base while preserving feasibility, we have to guarantee exchangeability with respect to all bases related to the exchanged elements, and such a property does not simply hold for a general flag base. Indeed, \cite[Section~1.9.2]{Borovik03} shows that flag matroids do not satisfy the strong exchange property in general.

\subsection*{Analysis via Bruhat Order}
As explained above, except in the case $s=1$, it seems difficult to reduce Flag-Matroid-CFR to a known combinatorial optimization problem. Here, using the notion of the Bruhat order from the theory of Coxeter groups (in particular, the symmetric group), we show that we can establish the optimality of the greedy solution in a simple way.

Let us identify $E$ with the set $[n]=\{1,2,\dots, n\}$. Then, the set of all rankings is nothing but the set $S_n$ of permutations on $[n]$, i.e., the symmetric group. Assume without loss of generality that the input ranking $\pi$ is the identity permutation $\id$, i.e., $\pi(i)=\id(i)=i$ for every $i\in[n]$. For any permutation (ranking) $\sigma\in S_n$, the Kendall tau distance $\dKT(\sigma, \id)$ coincides with the {\em inversion number} of $\sigma$, which is defined as
\[\inv(\sigma)=|\set{(i,j)\in [n]\times[n]| i<j,~\sigma(i)>\sigma(j)}|.\]
Thus, Flag-Matroid-CFR can be seen as the problem of finding $\sigma\in S_n$ minimizing $\inv(\sigma)$ subject to $(\Top_{k_1}(\sigma), \dots, \Top_{k_s}(\sigma))\in \F(\bM)$.

\medskip

We recall the Bruhat order (also called the strong Bruhat order) on $S_n$ following \cite[Sections 2.1 and 2.6]{BjornerBrenti}. For the Bruhat order for general Coxeter groups, see \cite[Chapter 2]{BjornerBrenti}.

Bruhat order is a partial order (i.e., a reflexive, antisymmetric, and transitive relation) on the elements of a Coxeter group.
Here is a definition of Bruhat order specialized to $S_n$, where we adopt a characterization given in \cite[Lemma~2.1.4]{BjornerBrenti} as a definition.
\begin{definition}[Bruhat Order on $S_n$]
Define a binary relation $\vartriangleleft$ on $S_n$ so that $\sigma \vartriangleleft \tau$ if and only if $\tau$ is obtained from $\sigma$ by swapping the entries at positions $i,j\in [n]$ such that $i<j$, $\sigma(i)<\sigma(j)$, and there is no $h$ such that $i<h<j$ and $\sigma(i)<\sigma(h)<\sigma(j)$.

The Bruhat order $\leq^{\rm B}$ on $S_n$ is the transitive and reflexive closure of $\vartriangleleft$.
\end{definition}
We write $\sigma<^{\rm B} \tau$ if $\sigma \leq^{\rm B} \tau$ and $\sigma\neq \tau$. Observe that $\sigma \vartriangleleft \tau$ implies $\inv(\tau)=\inv(\sigma)+1$. Indeed, let $i<j$ be the positions whose entries are swapped, as specified in the definition. Pairs of positions involving neither $i$ nor $j$ retain their inversion status. For any position $h\notin\{i,j\}$, the total contribution of the two pairs involving $h$ and either $i$ or $j$ remains unchanged after the swap: if $h<i$ or $h>j$, their contributions are interchanged; if $i<h<j$, the defining condition ensures that $\sigma(h)<\sigma(i)$ or $\sigma(h)>\sigma(j)$, so neither pair changes its inversion status. The pair $(i,j)$ changes from a noninversion to an inversion, giving a net increase of $1$ in the inversion number. Consequently, the Bruhat order on $S_n$ satisfies
\[\sigma<^{\rm B} \tau~~~\Longrightarrow~~~ \inv(\sigma)<\inv(\tau).\]

Therefore, if we can compute $\sigma^*\in S_n$ that is minimum with respect to the Bruhat order among all $\sigma\in S_n$ satisfying $(\Top_{k_1}(\sigma), \dots, \Top_{k_s}(\sigma))\in \F(\bM)$, then $\sigma^*$ is the optimal solution to Flag-Matroid-CFR. In fact, we can show this by using the following characterization of Bruhat order by Björner--Brenti~\cite{BjornerBrenti1996} (also presented in \cite[Theorem~2.6.3]{BjornerBrenti}).

\begin{theorem}[Tableau Criterion \cite{BjornerBrenti1996}]\label{thm:tableau}
For $\sigma,\tau\in S_n$, let $\sigma_{i,h}$ be the $i$th element in the increasing rearrangement of $\sigma(1), \sigma(2),\dots, \sigma(h)$ and similarly define $\tau_{i,h}$. Then, $\sigma\leq^{\rm B} \tau$ if and only if $\sigma_{i,h}\leq \tau_{i,h}$ for all $h\in D_R(\sigma)$ and $i\in[h]$, where $D_R(\sigma)=\set{h'\in[n-1]| \sigma(h')>\sigma(h'+1)}$.
\end{theorem}

Recall that, for a flag base $\bB=(B_1,\dots, B_s)\in \F(\bM)$, the ranking $\pi_{\bB}$ is defined by arranging $B_1$, $B_2 \setminus B_1$, $B_3 \setminus B_2$, $\dots$, $B_s \setminus B_{s-1}$, and $E \setminus B_s$ in this order, where the elements in each part are ordered according to $\pi$.
\begin{theorem}\label{thm:flag-CFR}
Given an instance $(\pi, \bM)$ of Flag-Matroid-CFR, suppose without loss of generality that $\pi=\id$. Let $\bB^*\in \F(\bM)$ be the unique $\preceq_{\pi}^{\mathrm G}$-minimum flag base. Then, for any other flag base $\bD\in \F(\bM)$, we have $\pi_{\bB^*}<^{\rm B} \pi_{\bD}$ (and hence $\inv(\pi_{\bB^*})<\inv(\pi_{\bD})$).
\end{theorem}
\begin{proof}
We write $\boldsymbol{k}=(k_1, k_2,\dots, k_s)$ for the rank of $\bM$. We let $\bB^*=(B_1^*,\dots,B_s^*)$ and $\bD=(D_1,\dots, D_s)$. Since $\pi=\id$, by the construction of $\pi_{\bB^*}$, when $\pi_{\bB^*}$ is regarded as a sequence, in each of the blocks $B_1^*,~ B_2^* \setminus B_1^*,~ \dots, B_s^* \setminus B_{s-1}^*,~ E \setminus B_s^*$, the elements are in increasing order, and hence $D_R(\pi_{\bB^*})\subseteq \{k_1, k_2,\dots, k_s\}$.

For each $h\in \{k_1, k_2,\dots, k_s\}$, we write $b_{i,h}$ (resp., $d_{i,h}$) for the $i$th element in the increasing rearrangement of the elements $\pi_{\bB^*}(1),\dots, \pi_{\bB^*}(h)$ (resp., $\pi_{\bD}(1),\dots, \pi_{\bD}(h)$). Note that, if $h=k_j$, then $\{\pi_{\bB^*}(1),\dots, \pi_{\bB^*}(h)\}=B_j^*$ and $\{\pi_{\bD}(1),\dots, \pi_{\bD}(h)\}=D_j$.
By Theorem~\ref{thm:tableau}, to show $\pi_{\bB^*}\leq^{\rm B} \pi_{\bD}$, it is sufficient to show $b_{i,h}\leq d_{i,h}$ for all $h\in D_R(\pi_{\bB^*})$ and $i\in [h]$. Since $D_R(\pi_{\bB^*})\subseteq \{k_1,k_2,\dots, k_s\}$, it is enough to show the stronger condition that $b_{i,h}\leq d_{i,h}$ for all $h\in\{k_1,k_2,\dots, k_s\}$ and $i\in [h]$. This condition is exactly equivalent to $\bB^* \preceq_{\pi}^{\mathrm G} \bD$, and thus it is satisfied by the definition of $\bB^*$. Therefore, $\pi_{\bB^*}\leq^{\rm B} \pi_{\bD}$. Since $\bD\neq \bB^*$, we have $\pi_{\bD}\neq \pi_{\bB^*}$, and hence $\pi_{\bB^*}<^{\rm B} \pi_{\bD}$.
\end{proof}

This implies the tractability of Flag-Matroid-CFR.
\begin{corollary}\label{cor:flag-CFR}
Flag-Matroid-CFR can be solved efficiently. The optimal solution is $\pi_{\bB^*}$, where $\bB^*$ is the unique $\preceq_{\pi}^{\mathrm G}$-minimum flag base.
\end{corollary}
We describe the construction of the optimal solution in Algorithm~\ref{alg:opt-ranking}.
Given an input ranking $\pi$ as an ordered list, the algorithm uses at most $ns$ independence oracle calls and $O(ns)$ additional time: $\bB^*$ is computed as in the proof of Proposition~\ref{prop:optimalflag}, and the elements are then placed into their respective blocks in the order in which they appear in $\pi$.

For the correctness of the algorithm, the flag matroid structure, i.e., the quotient relations between the matroids, is critical. See Appendix~\ref{app:relaxed} for a detailed discussion.

\begin{algorithm}[h]
\caption{Greedy Algorithm for Flag-Matroid-CFR}
\label{alg:opt-ranking}
\KwIn{A ranking $\pi$ and a flag matroid $\bM=(M_1,\dots,M_s)$ of rank $\boldsymbol{k}=(k_1,\dots,k_s)$.}
\KwOut{A ranking $\sigma$ minimizing $d_{\mathrm{KT}}(\sigma,\pi)$ subject to $\Top_{k_i}(\sigma)\in \B_i$ for all $i\in[s]$.}
\smallskip

\For{$i=1$ \KwTo $s$}{
    Compute the unique $\preceq_{\pi}^{\mathrm G}$-minimum base $B_i^*$ of $M_i$ by the greedy algorithm\;
}
Construct $\sigma$ by ordering the blocks $B_1^*$, $B_2^* \setminus B_1^*$, $\dots$, $B_s^* \setminus B_{s-1}^*$, and $E \setminus B_s^*$ in this order, with the elements in each block ordered according to $\pi$\;
\Return{$\sigma$}\;
\end{algorithm}

\section{Approximation Consequences for Rank Aggregation}\label{sec:approx}
Using the exact algorithms for Matroid-CFR and Flag-Matroid-CFR from Section~\ref{sec:matroid-CFR} as subroutines (in a generalized form for the former), we can extend the approximation frameworks in \cite{chak25improve} to the following settings while preserving the approximation ratios.

\begin{tcolorbox}[colback=white,colframe=black,boxrule=0.2pt,arc=0pt]
\textsc{Matroid Feasible Rank Aggregation (Matroid-FRA)}\\[2pt]
\textbf{Input:} $m$ rankings $\pi_1,\dots,\pi_m$ over $E$, and a matroid $M=(E,\B)$ of rank $k$.\\
\textbf{Output:} A ranking $\sigma$ over $E$ minimizing $\sum_{i=1}^m d_{\mathrm{KT}}(\sigma,\pi_i)$ subject to $\Top_k(\sigma)\in \B$.
\end{tcolorbox}

\begin{tcolorbox}[colback=white,colframe=black,boxrule=0.2pt,arc=0pt]
\textsc{Flag Matroid Feasible Rank Aggregation (Flag-Matroid-FRA)}\\[2pt]
\textbf{Input:} $m$ rankings $\pi_1,\dots,\pi_m$ over $E$, and a flag matroid $\bM=(M_1, \dots, M_s)$ of rank $\boldsymbol{k}=(k_1, \dots, k_s)$, where the base family of $M_i$ is denoted by $\B_i$.\\
\textbf{Output:} A ranking $\sigma$ over $E$ minimizing $\sum_{i=1}^m d_{\mathrm{KT}}(\sigma,\pi_i)$ subject to $\Top_{k_i}(\sigma)\in \B_i$ for each $i=1,2,\dots, s$.
\end{tcolorbox}

\begin{restatable}[$\star$]{theorem}{ThmApproxMat}\label{thm:approx-mat}
For any constant $\varepsilon>0$, there exists a $(2+\varepsilon)$-approximation algorithm for Matroid-FRA.
\end{restatable}

\begin{restatable}[$\star$]{theorem}{ThmApproxFlag}\label{thm:approx-flag}
There exists a $2.881$-approximation algorithm for Flag-Matroid-FRA.
\end{restatable}
The $(2+\varepsilon)$-approximation for Matroid-FRA is obtained by replacing the colorful bi-partition subproblem in their framework with a matroidal analogue and then applying the PTAS by Kenyon-Mathieu and Schudy~\cite{KenyonMathieuSchudy07} to the two resulting parts. The $2.881$-approximation for Flag-Matroid-FRA follows from their generic black-box framework once Flag-Matroid-CFR is polynomial-time solvable. We provide the details in Appendix~\ref{app:approx}.

\section{Hardness of Rank Aggregation with Few Voters}\label{sec:hardness}
In the classical Kemeny consensus (i.e., rank aggregation without constraints), the complexity of the problem with a fixed number $m$ of voters has been investigated. While the case with $m\in\{1,2\}$ is tractable, all other cases are NP-hard \cite{dwor01web,bach19kmajor,Madarasi2026, Peters2026}.

We show that Matroid-FRA with $m$ voters is NP-hard for any constant $m\geq 2$.
Note that the case $m=1$ coincides with Matroid-CFR and is thus solved in Theorem~\ref{thm:CFR}.
To obtain the hardness result for every fixed $m\geq 2$, it suffices
to prove hardness for $m=2$ and $m=3$, because adding any ranking and its reverse
increases the number of voters by two without changing the set
of optimal feasible rankings.

\begin{restatable}[$\star$]{theorem}{ThmHardTwo}\label{thm:two}
Matroid-FRA with $2$ voters is NP-hard.
\end{restatable}

\begin{restatable}[$\star$]{theorem}{ThmHardThree}\label{thm:three}
Matroid-FRA with $3$ voters is NP-hard.
\end{restatable}

The hardness with $3$ voters also follows from the subsequent results of Madarasi~\cite{Madarasi2026} and Peters~\cite{Peters2026}, by taking the constraint matroid to be a uniform matroid. We retain our original direct reduction in Appendix~\ref{app:hardness}. Our reduction additionally shows hardness even for the case in which each of the three input rankings is optimal for unconstrained Kemeny aggregation. The proofs for both cases $m=2$ and $m=3$, which are given in Appendix~\ref{app:hardness}, use simple partition matroids, and the statement below follows.

\begin{restatable}[$\star$]{corollary}{CorAll}\label{cor:all_m}
For any $m\geq 2$, Matroid-FRA with $m$ voters is NP-hard. The statement holds even when matroids are restricted to partition matroids.
\end{restatable}

\section*{Acknowledgments}
The authors thank Yuya Ryuzaki for helpful discussions at an early stage of this work, and the anonymous reviewers for their valuable feedback.
The second author was supported by JST ERATO Grant Number JPMJER2301, JST CRONOS Grant Number JPMJCS24K2, and JSPS KAKENHI Grant Number JP26K14706.

\bibliography{bib2doi}

@book{garey1979,
  author = {Michael R. Garey and David S. Johnson},
  title = {Computers and Intractability: A Guide to the Theory of NP-Completeness},
  publisher = {W. H. Freeman and Company},
  address = {San Francisco},
  year = {1979},
  doi = {10.5555/578533},
  url = {https://dl.acm.org/doi/10.5555/578533},
}

@incollection{karp1972,
  author = {Richard M. Karp},
  title = {Reducibility Among Combinatorial Problems},
  booktitle = {Complexity of Computer Computations},
  pages = {85--103},
  publisher = {Plenum Press},
  address = {New York},
  year = {1972},
  doi = {10.1007/978-1-4684-2001-2_9},
  url = {https://doi.org/10.1007/978-1-4684-2001-2_9},
  timestamp = {Wed, 16 Mar 2022 00:00:00 +0100},
  biburl = {https://dblp.org/rec/conf/coco/Karp72.bib},
  bibsource = {dblp computer science bibliography, https://dblp.org},
  _bib2doi_selected = {dblp:/rec/conf/coco/Karp72.bib},
  _bib2doi_confirmed = {true},
}

@article{hema05kemeny,
  author = {Edith Hemaspaandra and Holger Spakowski and Jörg Vogel},
  title = {The complexity of {Kemeny} elections},
  journal = {Theoretical Computer Science},
  volume = {349},
  pages = {382--391},
  year = {2005},
  doi = {10.1016/j.tcs.2005.08.031},
  url = {https://doi.org/10.1016/j.tcs.2005.08.031},
  timestamp = {Wed, 17 Feb 2021 00:00:00 +0100},
  biburl = {https://dblp.org/rec/journals/tcs/HemaspaandraSV05.bib},
  bibsource = {dblp computer science bibliography, https://dblp.org},
  _bib2doi_selected = {dblp:/rec/journals/tcs/HemaspaandraSV05.bib},
  _bib2doi_confirmed = {true},
}

@article{ailo08agg,
  author = {Nir Ailon and Moses Charikar and Alantha Newman},
  title = {Aggregating inconsistent information: Ranking and clustering},
  journal = {Journal of the ACM},
  volume = {55},
  number = {5},
  pages = {23:1--23:27},
  articleno = {23},
  year = {2008},
  doi = {10.1145/1411509.1411513},
  url = {https://doi.org/10.1145/1411509.1411513},
  timestamp = {Sat, 30 Sep 2023 01:00:00 +0200},
  biburl = {https://dblp.org/rec/journals/jacm/AilonCN08.bib},
  bibsource = {dblp computer science bibliography, https://dblp.org},
  _bib2doi_selected = {dblp:/rec/journals/jacm/AilonCN08.bib},
  _bib2doi_confirmed = {true},
}

@inproceedings{wang24survey,
  author = {Siyi Wang and Qi Deng and Shiwei Feng and Hong Zhang and Chao Liang},
  title = {A Survey on Rank Aggregation},
  booktitle = {Proceedings of the 33rd International Joint Conference on Artificial Intelligence (IJCAI 2024)},
  pages = {8281--8289},
  publisher = {International Joint Conferences on Artificial Intelligence Organization},
  year = {2024},
  doi = {10.24963/ijcai.2024/915},
  url = {https://doi.org/10.24963/ijcai.2024/915},
  note = {Survey Track},
  timestamp = {Tue, 22 Oct 2024 01:00:00 +0200},
  biburl = {https://dblp.org/rec/conf/ijcai/WangDFZL24.bib},
  bibsource = {dblp computer science bibliography, https://dblp.org},
  _bib2doi_selected = {dblp:/rec/conf/ijcai/WangDFZL24.bib},
  _bib2doi_confirmed = {true},
}

@inproceedings{chak25improve,
  author = {Diptarka Chakraborty and Himika Das and Sanjana Dey and Alvin Hong Yao Yan},
  title = {Improved Rank Aggregation Under Fairness Constraint},
  booktitle = {Proceedings of the 34th International Joint Conference on Artificial Intelligence (IJCAI 2025)},
  pages = {330--338},
  publisher = {International Joint Conferences on Artificial Intelligence Organization},
  year = {2025},
  doi = {10.24963/ijcai.2025/38},
  url = {https://doi.org/10.24963/ijcai.2025/38},
  note = {Main Track},
  timestamp = {Tue, 24 Mar 2026 00:00:00 +0100},
  biburl = {https://dblp.org/rec/conf/ijcai/ChakrabortyDDY25.bib},
  bibsource = {dblp computer science bibliography, https://dblp.org},
  _bib2doi_selected = {dblp:/rec/conf/ijcai/ChakrabortyDDY25.bib},
  _bib2doi_confirmed = {true},
}

@inproceedings{chak22fair,
  author = {Diptarka Chakraborty and Syamantak Das and Arindam Khan and Aditya Subramanian},
  title = {Fair Rank Aggregation},
  booktitle = {Advances in Neural Information Processing Systems (NeurIPS 2022)},
  volume = {35},
  pages = {23965--23978},
  year = {2022},
  eprint = {2308.10499},
  archiveprefix = {arXiv},
  primaryclass = {cs.DS},
  url = {https://arxiv.org/abs/2308.10499},
  note = {Conference version in NeurIPS 2022; arXiv version cited in the main text},
  doi = {10.48550/arXiv.2308.10499},
}

@inproceedings{fitz21kemeny,
  author = {Zack Fitzsimmons and Edith Hemaspaandra},
  title = {{Kemeny} Consensus Complexity},
  booktitle = {Proceedings of the 30th International Joint Conference on Artificial Intelligence (IJCAI 2021)},
  pages = {196--202},
  publisher = {International Joint Conferences on Artificial Intelligence Organization},
  year = {2021},
  doi = {10.24963/ijcai.2021/28},
  url = {https://doi.org/10.24963/ijcai.2021/28},
  note = {Main Track},
  timestamp = {Wed, 25 Aug 2021 01:00:00 +0200},
  biburl = {https://dblp.org/rec/conf/ijcai/FitzsimmonsH21.bib},
  bibsource = {dblp computer science bibliography, https://dblp.org},
  _bib2doi_selected = {dblp:/rec/conf/ijcai/FitzsimmonsH21.bib},
  _bib2doi_confirmed = {true},
}

@article{krai23popular,
  author = {Sonja Kraiczy and Ágnes Cseh and David Manlove},
  title = {On weakly and strongly popular rankings},
  journal = {Discrete Applied Mathematics},
  volume = {340},
  pages = {134--152},
  year = {2023},
  doi = {10.1016/j.dam.2023.06.041},
  url = {https://doi.org/10.1016/j.dam.2023.06.041},
  timestamp = {Wed, 01 Nov 2023 00:00:00 +0100},
  biburl = {https://dblp.org/rec/journals/dam/KraiczyCM23.bib},
  bibsource = {dblp computer science bibliography, https://dblp.org},
  _bib2doi_selected = {dblp:/rec/journals/dam/KraiczyCM23.bib},
  _bib2doi_confirmed = {true},
}

@inproceedings{dwor01web,
  author = {Cynthia Dwork and Ravi Kumar and Moni Naor and Dandapani Sivakumar},
  title = {Rank aggregation methods for the {Web}},
  booktitle = {Proceedings of the 10th International Conference on World Wide Web (WWW 2001)},
  pages = {613--622},
  publisher = {ACM},
  year = {2001},
  doi = {10.1145/371920.372165},
  url = {https://doi.org/10.1145/371920.372165},
  timestamp = {Mon, 08 Mar 2021 00:00:00 +0100},
  biburl = {https://dblp.org/rec/conf/www/DworkKNS01.bib},
  bibsource = {dblp computer science bibliography, https://dblp.org},
  _bib2doi_selected = {dblp:/rec/conf/www/DworkKNS01.bib},
  _bib2doi_confirmed = {true},
}

@article{bach19kmajor,
  author = {Georg Bachmeier and Felix Brandt and Christian Geist and Paul Harrenstein and Keyvan Kardel and Dominik Peters and Hans Georg Seedig},
  title = {$k$-Majority digraphs and the hardness of voting with a constant number of voters},
  journal = {Journal of Computer and System Sciences},
  volume = {105},
  pages = {130--157},
  year = {2019},
  doi = {10.1016/j.jcss.2019.04.005},
  url = {https://doi.org/10.1016/j.jcss.2019.04.005},
  timestamp = {Sun, 19 Jan 2025 00:00:00 +0100},
  biburl = {https://dblp.org/rec/journals/jcss/BachmeierBGHKPS19.bib},
  bibsource = {dblp computer science bibliography, https://dblp.org},
  _bib2doi_selected = {dblp:/rec/journals/jcss/BachmeierBGHKPS19.bib},
  _bib2doi_confirmed = {true},
}

@article{betz09fpt,
  author = {Nadja Betzler and Michael R. Fellows and Jiong Guo and Rolf Niedermeier and Frances A. Rosamond},
  title = {Fixed-parameter algorithms for {Kemeny} rankings},
  journal = {Theoretical Computer Science},
  volume = {410},
  pages = {4554--4570},
  year = {2009},
  doi = {10.1016/j.tcs.2009.08.033},
  url = {https://doi.org/10.1016/j.tcs.2009.08.033},
  timestamp = {Wed, 17 Feb 2021 00:00:00 +0100},
  biburl = {https://dblp.org/rec/journals/tcs/BetzlerFGNR09.bib},
  bibsource = {dblp computer science bibliography, https://dblp.org},
  _bib2doi_selected = {dblp:/rec/journals/tcs/BetzlerFGNR09.bib},
  _bib2doi_confirmed = {true},
}

@article{Kem59,
  author = {John G. Kemeny},
  title = {Mathematics without Numbers},
  journal = {Daedalus},
  volume = {88},
  number = {4},
  pages = {577--591},
  year = {1959},
}

@article{You88,
  author = {Hobart Peyton Young},
  title = {Condorcet's Theory of Voting},
  journal = {American Political Science Review},
  volume = {82},
  number = {4},
  pages = {1231--1244},
  year = {1988},
  doi = {10.2307/1961757},
  url = {https://doi.org/10.2307/1961757},
}

@inproceedings{ZehlikeBCM17,
  author = {Meike Zehlike and Francesco Bonchi and Carlos Castillo and Sara Hajian and Mohamed Megahed and Ricardo Baeza-Yates},
  title = {FA*IR: A Fair Top-k Ranking Algorithm},
  booktitle = {Proceedings of the 26th ACM International Conference on Information and Knowledge Management (CIKM 2017)},
  pages = {1569--1578},
  publisher = {ACM},
  year = {2017},
  doi = {10.1145/3132847.3132938},
  url = {https://doi.org/10.1145/3132847.3132938},
  timestamp = {Thu, 23 Jun 2022 01:00:00 +0200},
  biburl = {https://dblp.org/rec/conf/cikm/ZehlikeB0HMB17.bib},
  bibsource = {dblp computer science bibliography, https://dblp.org},
  _bib2doi_selected = {dblp:/rec/conf/cikm/ZehlikeB0HMB17.bib},
  _bib2doi_confirmed = {true},
}

@inproceedings{CelisSV18,
  author = {L. Elisa Celis and Damian Straszak and Nisheeth K. Vishnoi},
  title = {Ranking with Fairness Constraints},
  booktitle = {45th International Colloquium on Automata, Languages, and Programming (ICALP 2018)},
  series = {Leibniz International Proceedings in Informatics (LIPIcs)},
  volume = {107},
  pages = {{28:1--28:15}},
  publisher = {Schloss Dagstuhl -- Leibniz-Zentrum für Informatik},
  year = {2018},
  doi = {10.4230/LIPIcs.ICALP.2018.28},
  url = {https://drops.dagstuhl.de/entities/document/10.4230/LIPIcs.ICALP.2018.28},
  timestamp = {Thu, 23 Aug 2018 01:00:00 +0200},
  biburl = {https://dblp.org/rec/conf/icalp/CelisSV18.bib},
  bibsource = {dblp computer science bibliography, https://dblp.org},
  _bib2doi_selected = {dblp:/rec/conf/icalp/CelisSV18.bib},
  _bib2doi_confirmed = {true},
}

@book{Oxley11,
  author = {James Oxley},
  title = {Matroid Theory},
  publisher = {Oxford University Press},
  year = {2011},
  doi = {10.1093/acprof:oso/9780198566946.001.0001},
  url = {https://doi.org/10.1093/acprof:oso/9780198566946.001.0001},
  edition = {2nd},
}

@book{Schrijver03,
  author = {Alexander Schrijver},
  title = {Combinatorial Optimization: Polyhedra and Efficiency},
  publisher = {Springer},
  year = {2003},
}

@article{Gale68,
  author = {David Gale},
  title = {Optimal Assignments in an Ordered Set: An Application of Matroid Theory},
  journal = {Journal of Combinatorial Theory},
  volume = {4},
  number = {2},
  pages = {176--180},
  year = {1968},
  doi = {10.1016/S0021-9800(68)80039-0},
  url = {https://doi.org/10.1016/S0021-9800(68)80039-0},
}

@article{Vince02,
  author = {Andrew Vince},
  title = {A Framework for the Greedy Algorithm},
  journal = {Discrete Applied Mathematics},
  volume = {121},
  number = {1--3},
  pages = {247--260},
  year = {2002},
  doi = {10.1016/S0166-218X(01)00362-6},
  url = {https://doi.org/10.1016/S0166-218X(01)00362-6},
  timestamp = {Thu, 23 Jun 2022 01:00:00 +0200},
  biburl = {https://dblp.org/rec/journals/dam/Vince02.bib},
  bibsource = {dblp computer science bibliography, https://dblp.org},
  _bib2doi_selected = {dblp:/rec/journals/dam/Vince02.bib},
  _bib2doi_confirmed = {true},
}

@article{JarraLorscheid24,
  author = {Manoel Jarra and Oliver Lorscheid},
  title = {Flag matroids with coefficients},
  journal = {Advances in Mathematics},
  volume = {436},
  pages = {109396},
  year = {2024},
  doi = {10.1016/j.aim.2023.109396},
  url = {https://doi.org/10.1016/j.aim.2023.109396},
}

@article{HorschSzigeti21,
  author = {Florian Hörsch and Zoltán Szigeti},
  title = {Packing of mixed hyperarborescences with flexible roots via matroid intersection},
  journal = {The Electronic Journal of Combinatorics},
  volume = {28},
  number = {3},
  pages = {P3.29},
  year = {2021},
  doi = {10.37236/10105},
  url = {https://doi.org/10.37236/10105},
  timestamp = {Mon, 01 Dec 2025 00:00:00 +0100},
  biburl = {https://dblp.org/rec/journals/combinatorics/HorschS21.bib},
  bibsource = {dblp computer science bibliography, https://dblp.org},
  _bib2doi_selected = {dblp:/rec/journals/combinatorics/HorschS21.bib},
  _bib2doi_confirmed = {true},
}

@article{yoko17gpoly,
  author = {Yu Yokoi},
  title = {A Generalized Polymatroid Approach to Stable Matchings with Lower Quotas},
  journal = {Mathematics of Operations Research},
  volume = {42},
  number = {1},
  pages = {238--255},
  year = {2017},
  doi = {10.1287/moor.2016.0802},
  url = {https://doi.org/10.1287/moor.2016.0802},
  timestamp = {Sun, 28 May 2017 01:00:00 +0200},
  biburl = {https://dblp.org/rec/journals/mor/Yokoi17.bib},
  bibsource = {dblp computer science bibliography, https://dblp.org},
  _bib2doi_selected = {dblp:/rec/journals/mor/Yokoi17.bib},
  _bib2doi_confirmed = {true},
}

@incollection{CameronDMS22,
  author = {Amanda Cameron and Rodica Dinu and Mateusz Michałek and Tim Seynnaeve},
  title = {Flag matroids: algebra and geometry},
  booktitle = {Interactions with Lattice Polytopes},
  pages = {73--114},
  publisher = {Springer International Publishing},
  address = {Cham},
  year = {2022},
  doi = {10.1007/978-3-030-98327-7_4},
  url = {https://doi.org/10.1007/978-3-030-98327-7_4},
}

@book{Borovik03,
  author = {Alexandre V. Borovik and Israel M. Gelfand and Neil White},
  title = {Coxeter Matroids},
  series = {Progress in Mathematics},
  volume = {216},
  publisher = {Birkh\"auser Boston},
  address = {Boston, MA},
  year = {2003},
  doi = {10.1007/978-1-4612-2066-4},
  url = {https://doi.org/10.1007/978-1-4612-2066-4},
}

@article{brandenburg2024quotients,
  author = {Marie-Charlotte Brandenburg and Georg Loho and Ben Smith},
  title = {Quotients of {M}-convex Sets and {M}-convex Functions},
  journal = {Combinatorial Theory},
  volume = {6},
  number = {1},
  year = {2026},
  doi = {10.5070/C66165701},
  url = {https://doi.org/10.5070/C66165701},
  timestamp = {Sun, 24 May 2026 01:00:00 +0200},
  biburl = {https://dblp.org/rec/journals/combtheory/BrandenburgLS26.bib},
  bibsource = {dblp computer science bibliography, https://dblp.org},
  _bib2doi_selected = {dblp:/rec/journals/combtheory/BrandenburgLS26.bib},
  _bib2doi_confirmed = {true},
}

@book{Frankbook,
  author = {Andr{\'a}s Frank},
  title = {Connections in Combinatorial Optimization},
  series = {Oxford Lecture Series in Mathematics and its Applications},
  volume = {38},
  publisher = {Oxford University Press},
  address = {Oxford},
  year = {2011},
  isbn = {978-0-19-920527-1},
}

@book{BjornerBrenti,
  author = {Anders Björner and Francesco Brenti},
  title = {Combinatorics of Coxeter Groups},
  series = {Graduate Texts in Mathematics},
  volume = {231},
  publisher = {Springer},
  year = {2005},
  doi = {10.1007/3-540-27596-7},
  url = {https://doi.org/10.1007/3-540-27596-7},
}

@article{BjornerBrenti1996,
  author = {Anders Björner and Francesco Brenti},
  title = {An improved tableau criterion for {B}ruhat order},
  journal = {The Electronic Journal of Combinatorics},
  volume = {3},
  number = {1},
  pages = {R22},
  year = {1996},
  doi = {10.37236/1246},
  url = {https://doi.org/10.37236/1246},
  timestamp = {Mon, 01 Dec 2025 00:00:00 +0100},
  biburl = {https://dblp.org/rec/journals/combinatorics/BjornerB96.bib},
  bibsource = {dblp computer science bibliography, https://dblp.org},
  _bib2doi_selected = {dblp:/rec/journals/combinatorics/BjornerB96.bib},
  _bib2doi_confirmed = {true},
}

@article{BorovikGelfandStone02,
  author = {Alexandre V. Borovik and Israel M. Gelfand and David A. Stone},
  title = {On the Topology of the Combinatorial Flag Varieties},
  journal = {Discrete \& Computational Geometry},
  volume = {27},
  number = {2},
  pages = {195--214},
  year = {2002},
  doi = {10.1007/s00454-001-0061-8},
  url = {https://doi.org/10.1007/s00454-001-0061-8},
  timestamp = {Thu, 12 Mar 2020 00:00:00 +0100},
  biburl = {https://dblp.org/rec/journals/dcg/BorovikGS02.bib},
  bibsource = {dblp computer science bibliography, https://dblp.org},
  _bib2doi_selected = {dblp:/rec/journals/dcg/BorovikGS02.bib},
  _bib2doi_confirmed = {true},
}

@incollection{Bry86,
  author = {Thomas Brylawski},
  editor = {White, Neil},
  title = {Constructions},
  booktitle = {Theory of Matroids},
  series = {Encyclopedia of Mathematics and its Applications},
  volume = {26},
  pages = {127--223},
  publisher = {Cambridge University Press},
  address = {Cambridge},
  year = {1986},
  doi = {10.1017/CBO9780511629563.010},
  url = {https://doi.org/10.1017/CBO9780511629563.010},
}

@article{BartholdiTT89,
  author = {Bartholdi, III, John J. and Craig A. Tovey and Michael A. Trick},
  title = {Voting schemes for which it can be difficult to tell who won the election},
  journal = {Social Choice and Welfare},
  volume = {6},
  number = {2},
  pages = {157--165},
  year = {1989},
  doi = {10.1007/BF00303169},
  url = {https://doi.org/10.1007/BF00303169},
}

@inproceedings{KenyonMathieuSchudy07,
  author = {Claire Kenyon{-}Mathieu and Warren Schudy},
  title = {How to rank with few errors},
  booktitle = {Proceedings of the 39th Annual ACM Symposium on Theory of Computing (STOC 2007)},
  pages = {95--103},
  publisher = {ACM},
  year = {2007},
  doi = {10.1145/1250790.1250806},
  url = {https://doi.org/10.1145/1250790.1250806},
  timestamp = {Wed, 14 Nov 2018 00:00:00 +0100},
  biburl = {https://dblp.org/rec/conf/stoc/Kenyon-MathieuS07.bib},
  bibsource = {dblp computer science bibliography, https://dblp.org},
  _bib2doi_selected = {dblp:/rec/conf/stoc/Kenyon-MathieuS07.bib},
  _bib2doi_confirmed = {true},
}

@article{NaturalFlag,
  author = {Anna de Mier},
  title = {A natural family of flag matroids},
  journal = {SIAM Journal on Discrete Mathematics},
  volume = {21},
  number = {1},
  pages = {130--140},
  year = {2007},
  doi = {10.1137/050627873},
  url = {https://doi.org/10.1137/050627873},
  timestamp = {Sat, 25 Apr 2020 01:00:00 +0200},
  biburl = {https://dblp.org/rec/journals/siamdm/Mier07.bib},
  bibsource = {dblp computer science bibliography, https://dblp.org},
  _bib2doi_selected = {dblp:/rec/journals/siamdm/Mier07.bib},
  _bib2doi_confirmed = {true},
}

@article{YoungLevenglick78,
  author = {H. Peyton Young and Arthur Levenglick},
  title = {A consistent extension of {Condorcet}'s election principle},
  journal = {SIAM Journal on Applied Mathematics},
  volume = {35},
  number = {2},
  pages = {285--300},
  year = {1978},
  doi = {10.1137/0135023},
  url = {https://doi.org/10.1137/0135023},
}

@article{brualdi71,
  author = {Richard A. Brualdi},
  title = {Induced matroids},
  journal = {Proceedings of the American Mathematical Society},
  volume = {29},
  number = {2},
  pages = {213--221},
  year = {1971},
  doi = {10.2307/2038115},
  url = {https://doi.org/10.2307/2038115},
}

@article{perfect69,
  author = {Hazel Perfect},
  title = {Independence spaces and combinatorial problems},
  journal = {Proceedings of the London Mathematical Society},
  volume = {s3-19},
  number = {1},
  pages = {17--30},
  year = {1969},
  doi = {10.1112/plms/s3-19.1.17},
  url = {https://doi.org/10.1112/plms/s3-19.1.17},
}

@inproceedings{wei22sigmod,
  author = {Dong Wei and Md Mouinul Islam and Baruch Schieber and Senjuti {Basu Roy}},
  title = {Rank aggregation with proportionate fairness},
  booktitle = {Proceedings of the 2022 International Conference on Management of Data (SIGMOD 2022)},
  pages = {262--275},
  publisher = {ACM},
  year = {2022},
  doi = {10.1145/3514221.3517865},
  url = {https://doi.org/10.1145/3514221.3517865},
  timestamp = {Tue, 14 Jun 2022 01:00:00 +0200},
  biburl = {https://dblp.org/rec/conf/sigmod/0001ISR22.bib},
  bibsource = {dblp computer science bibliography, https://dblp.org},
  _bib2doi_selected = {dblp:/rec/conf/sigmod/0001ISR22.bib},
  _bib2doi_confirmed = {true},
}

@inproceedings{BredereckFILS18,
  author = {Robert Bredereck and Piotr Faliszewski and Ayumi Igarashi and Martin Lackner and Piotr Skowron},
  title = {Multiwinner Elections with Diversity Constraints},
  booktitle = {Proceedings of the 32nd AAAI Conference on Artificial Intelligence (AAAI 2018)},
  pages = {933--940},
  year = {2018},
  publisher = {AAAI Press},
  url = {https://ojs.aaai.org/index.php/AAAI/article/view/11457},
  timestamp = {Mon, 04 Sep 2023 01:00:00 +0200},
  biburl = {https://dblp.org/rec/conf/aaai/BredereckFILS18.bib},
  bibsource = {dblp computer science bibliography, https://dblp.org},
  doi = {10.1609/aaai.v32i1.11457},
  _bib2doi_selected = {dblp:/rec/conf/aaai/BredereckFILS18.bib},
  _bib2doi_confirmed = {true},
}

@inproceedings{GorantlaMDL23,
  author = {Sruthi Gorantla and Anay Mehrotra and Amit Deshpande and Anand Louis},
  title = {Sampling Individually-Fair Rankings that are Always Group Fair},
  booktitle = {Proceedings of the 2023 AAAI/ACM Conference on AI, Ethics, and Society (AIES '23)},
  pages = {205--216},
  year = {2023},
  publisher = {Association for Computing Machinery},
  doi = {10.1145/3600211.3604671},
  url = {https://doi.org/10.1145/3600211.3604671},
  timestamp = {Sun, 24 Sep 2023 01:00:00 +0200},
  biburl = {https://dblp.org/rec/conf/aies/GorantlaMDL23.bib},
  bibsource = {dblp computer science bibliography, https://dblp.org},
  _bib2doi_selected = {dblp:/rec/conf/aies/GorantlaMDL23.bib},
  _bib2doi_confirmed = {true},
}

@article{DoganImamuraYenmez25,
  author = {Battal Do{\u{g}}an and Kenzo Imamura and M. Bumin Yenmez},
  title = {Market Design with Deferred Acceptance: A Recipe for Characterizations},
  journal = {Journal of Economic Theory},
  volume = {228},
  pages = {106057},
  year = {2025},
  doi = {10.1016/j.jet.2025.106057},
  url = {https://doi.org/10.1016/j.jet.2025.106057},
  timestamp = {Sun, 02 Nov 2025 00:00:00 +0100},
  biburl = {https://dblp.org/rec/journals/jet/DoganIY25.bib},
  bibsource = {dblp computer science bibliography, https://dblp.org},
  _bib2doi_selected = {dblp:/rec/journals/jet/DoganIY25.bib},
  _bib2doi_confirmed = {true},
}

@misc{Madarasi2026,
  author        = {P{\'e}ter Madarasi},
  title         = {The Complexity of {Kemeny} Aggregation with Three Rankings},
  year          = {2026},
  eprint        = {2607.28588},
  archivePrefix = {arXiv},
  primaryClass  = {cs.GT},
  doi           = {10.48550/arXiv.2607.28588},
  url           = {https://arxiv.org/abs/2607.28588}
}

@misc{Peters2026,
  author        = {Dominik Peters},
  title         = {{Kemeny} Rank Aggregation is {NP}-Hard for Three Voters},
  year          = {2026},
  eprint        = {2607.25540},
  archivePrefix = {arXiv},
  primaryClass  = {cs.GT},
  doi           = {10.48550/arXiv.2607.25540},
  url           = {https://arxiv.org/abs/2607.25540}
}
\clearpage
\appendix

\section{A Note on Exchange-Based Proofs for Block-Fairness}\label{app:exchange}

In \cite{chak22fair}, the authors give an exact algorithm for Closest Fair Ranking under block-$k$-fairness constraints. When specialized to this setting, our algorithm for Flag-Matroid-CFR produces the same ranking as theirs.

They establish optimality by an exchange argument, aiming to show that a ranking different from the algorithm's output admits a feasible exchange of two elements that improves the objective value. However, to the best of our understanding, the exchange step as stated requires an additional justification: the proposed swap need not preserve all block-prefix constraints simultaneously. We explain the issue in \cite[Claims~3.8 and~3.9]{chak22fair}, using the notation of the present paper.

The key point is that block-$k$-fairness is a \emph{simultaneous} feasibility condition:
By its definition, a block-$k$-fair ranking must satisfy the fairness constraint for every block-prefix of length at least~$k$.
By contrast, the swap argument proved in their Theorem~3.2 concerns fairness for one fixed prefix length only. Thus, when Theorem~3.2 is invoked
inside their Claim~3.8 to conclude that swapping two elements produces another
block-$k$-fair ranking, this conclusion does not follow from Theorem~3.2 alone.

This issue can already be seen in a small example. Consider the groups
\[
G_1=\{a_1,a_2\},\qquad G_2=\{b_1,b_2\},\qquad G_3=\{c\}
\]
with block length $b=2$ and $k=2$, and let
\[
\alpha_1=\alpha_2=\alpha_3=0,\qquad \beta_1=\tfrac12,\qquad \beta_2=\beta_3=1.
\]
Then block-$k$-fairness requires that the top-$2$ prefix contain at most one element of $G_1$, and the top-$4$ prefix contain at most two elements of $G_1$. There are no other constraints, since $\alpha_1=\alpha_2=\alpha_3=0$ and $\beta_2=\beta_3=1$ impose nothing.

Now consider the input ranking
\[
\pi=(a_1,a_2,b_1,b_2,c).
\]
Then, the output of the algorithm in \cite{chak22fair} (which coincides with our algorithm's output) is
\[
\sigma=(a_1,b_1,a_2,b_2,c).
\]
Consider a block-$k$-fair ranking
\[
\pi^*=(a_1,c,b_1,b_2,a_2).
\]
This satisfies block-$k$-fairness and also preserves the intra-group orderings (i.e., the input ordering in each group). Since $\Top_{4}(\sigma)\neq \Top_{4}(\pi^*)$, following the argument in \cite[Claim~3.8]{chak22fair}, take the elements $a_2\in \Top_{4}(\sigma)\setminus \Top_{4}(\pi^*)$ and $c\in \Top_{4}(\pi^*)\setminus \Top_{4}(\sigma)$. The intended argument is that, by swapping $a_2$ and $c$ in $\pi^*$, one obtains a block-$k$-fair ranking $\pi^{**}$ that is closer to $\pi$ than $\pi^*$. As we have $a_2\prec_{\pi} c$, swapping $a_2$ and $c$ indeed reduces the objective value. However, the resulting ranking
\[
\pi^{**}=(a_1,a_2,b_1,b_2,c)
\]
contains more than one element of $G_1$ in its top-$2$ prefix, violating block-$k$-fairness.

Therefore, the step in Claim~3.8 asserting that the exchanged ranking is still
block-$k$-fair does not seem to follow as stated without an additional argument. The same concern also applies to Claim~3.9. Thus, the exchange-based proof of Theorem~3.4 appears to require an additional justification at this step.

Importantly, Claims~3.8 and~3.9 and Theorem~3.4 themselves are true; our proof based on Bruhat order establishes their correctness.

\section{Failure of the Greedy Approach without Quotient Relations}\label{app:relaxed}
As a generalization of block-$k$-fairness, we consider flag matroid constraints on rankings. However, for a sequence of matroids, one can consider weaker conditions, and it is natural to ask to what extent the tractability of Flag-Matroid-CFR extends. We do not know a clear answer, but here we provide some observations.

Suppose that we have a sequence $(M_1, M_2,\dots, M_s)$ of matroids on the same ground set $E$ with ranks $k_1<k_2<\cdots <k_s$.
Each matroid is given as $M_i=(E, \B_i)$.

Note that, if we have no additional assumption at all, then even finding a feasible ranking is NP-hard when $s\geq 3$. This follows easily from the fact that 3-matroid intersection is NP-hard.

We then consider imposing some assumptions on $(M_1,M_2,\dots, M_s)$ to obtain a reasonable class. As gradual restrictions, we can consider the following three settings:

\begin{enumerate}
\item The independent set families satisfy $\I_1\subseteq \I_2\subseteq \cdots\subseteq \I_s$. This is equivalent to the condition that for every $i\in[s-1]$ and $B\in\B_i$, there exists $B'\in \B_{i+1}$ with $B\subseteq B'$.
\item The above condition holds, and in addition, for every $i\in[s]\setminus \{1\}$ and $B\in\B_i$, there exists $B'\in \B_{i-1}$ with $B'\subseteq B$.
\item $(M_1,M_2,\dots, M_s)$ forms a flag matroid.
\end{enumerate}

Clearly, 2 implies 1. It is also not hard to see that 3 implies 2; a quotient relation implies $\rk_{M_{i}}(X)-\rk_{M_{i}}(\emptyset)\leq \rk_{M_{i+1}}(X)-\rk_{M_{i+1}}(\emptyset)$ for every $X\subseteq E$, which implies $\I_i\subseteq \I_{i+1}$. In addition, for any $B\in \B_i$, we have $\rk_{M_{i-1}}(E)-\rk_{M_{i-1}}(B)\leq \rk_{M_{i}}(E)-\rk_{M_{i}}(B)=0$, which implies $\rk_{M_{i-1}}(B)=\rk_{M_{i-1}}(E)$ and hence the existence of $B'\in \B_{i-1}$ with $B'\subseteq B$.

There are indeed gaps between any two of these three classes.
An example satisfying 1 but violating 2 is given by $\B_1=\{\{1,2\}\}$ and $\B_2=\{\{1,2,3\}, \{1,2,4\},\{1,3,4\}\}$.
An example of a matroid pair satisfying 2 but violating 3 is given in Figure~\ref{fig:greedy}.

Given an input ranking $\pi$, consider the following two greedy algorithms.

\medskip
\noindent{\bf Forward Greedy Algorithm}
\begin{enumerate}
\item Let $B_0=\emptyset$ and for $i=1,2,\dots, s$, do the following.
\begin{itemize}
\item Let $X\gets B_{i-1}$.
\item For each $e\in E\setminus B_{i-1}$, in the ascending order of $\pi$, if $X\cup \{e\}\in \I_i$, then $X\gets X\cup \{e\}$.
\item Set $B_i\gets X$.
\end{itemize}
\item Let $\bB=(B_1, \dots, B_s)$ and output $\pi_{\bB}$ as defined in Section~\ref{sec:matroid-CFR}.
\end{enumerate}

\smallskip
\noindent{\bf Backward Greedy Algorithm}
\begin{enumerate}
\item Let $B_{s+1}= E$ and for $i=s, s-1,\dots, 1$, do the following.
\begin{itemize}
\item Let $X\gets \emptyset$.
\item For each $e\in B_{i+1}$, in the ascending order of $\pi$, if $X\cup \{e\}\in \I_i$, then $X\gets X\cup \{e\}$.
\item Set $B_i\gets X$.
\end{itemize}
\item Let $\bB=(B_1, \dots, B_s)$ and output $\pi_{\bB}$ as defined in Section~\ref{sec:matroid-CFR}.
\end{enumerate}

Then the following hold.

\begin{itemize}
\item In setting 3, i.e., when the matroids form a flag matroid, both algorithms output the same ranking, and it is the unique optimal solution to Flag-Matroid-CFR.

\item In setting 2, both algorithms always output feasible rankings, but it is possible that neither output is optimal. We demonstrate this below.

\item In setting 1, the Forward Greedy Algorithm always outputs a feasible ranking, while the backward one may not.
\end{itemize}

We provide an example that satisfies the conditions in setting 2 and on which neither of the two greedy algorithms finds an optimal solution.

Consider a pair of matroids $M_1$ and $M_2$ on $E=\{1,2,\dots, 12\}$, which are both partition matroids defined as in Figure~\ref{fig:greedy}.
The smaller matroid $M_1$ has rank $6$ and is defined by the partition drawn in blue. The numbers written in blue show the upper quotas of those classes. The larger matroid $M_2$ has rank $8$, and its partition classes are given by the red dashed regions, whose upper quotas are also shown.
This pair belongs to setting 2, while it does not fit setting 3, since $\rk_{M_1}(\{1,2,5\})-\rk_{M_1}(\{1,2\})=1>0=\rk_{M_2}(\{1,2,5\})-\rk_{M_2}(\{1,2\})$.

\smallskip
Suppose that the input ranking $\pi$ is the identity permutation.

We can observe that the output $\sigma_{\rm f}$ of the forward greedy algorithm and the output $\sigma_{\rm b}$ of the backward greedy algorithm are, respectively,
\begin{align*}
\sigma_{\rm f}&= (1,4,5,7,9,10,6,12,2,3,8,11),\\
\sigma_{\rm b}&= (1,4,6,7,10,12,2,8,3,5,9,11).
\end{align*}
The Kendall tau distances to $\pi$ are $\dKT(\sigma_{\rm f},\pi)=21$ and $\dKT(\sigma_{\rm b},\pi)=21$. However, the feasible ranking \[\sigma^*\coloneqq (1,4,6,7,9,10,2,12,3,5,8,11)\]
has a smaller Kendall tau distance $\dKT(\sigma^*,\pi)=20$.

Thus, neither of the two greedy algorithms finds an optimal solution on this example.

\begin{figure}[h]
    \centering
        \begin{tikzpicture}[x=1.6cm, y=1.2cm, element/.style={circle, draw=black, thick, fill=white, minimum size=18pt, inner sep=0pt, font=\bfseries}]

        \node[element] (n1) at (0, 1) {\(1\)};
        \node[element] (n5) at (1, 1) {\(5\)};
        \node[element] (n4) at (2, 1) {\(4\)};
        \node[element] (n2) at (0, 0) {\(2\)};
        \node[element] (n3) at (1, 0) {\(3\)};
        \node[element] (n6) at (2, 0) {\(6\)};
        \node[element] (n7) at (3, 1) {\(7\)};
        \node[element] (n9) at (4, 1) {\(9\)};
        \node[element] (n10) at (5, 1) {\(10\)};
        \node[element] (n8) at (3, 0) {\(8\)};
        \node[element] (n11) at (4, 0) {\(11\)};
        \node[element] (n12) at (5, 0) {\(12\)};

        \begin{pgfonlayer}{background}
        \draw[fill, blue!20] \convexpath{n1, n3, n2}{12pt};
        \draw[fill, blue!20] \convexpath{n4, n6, n5}{12pt};
        \draw[fill, blue!20] \convexpath{n7, n11, n8}{12pt};
        \draw[fill, blue!20] \convexpath{n10, n12, n9}{12pt};
        \node[draw=red!70!gray, fill=red!30, fill opacity = 0.4, thick, dashed, rounded corners=6pt, fit=(n1)(n5)(n2)(n3), inner sep=3pt] (m1_1) {};
        \node[draw=red!70!gray, fill=red!30, fill opacity = 0.4, thick, dashed, rounded corners=6pt, fit=(n4)(n6), inner sep=3pt] (m1_2) {};
        \node[draw=red!70!gray, fill=red!30, fill opacity = 0.4, thick, dashed, rounded corners=6pt, fit=(n7)(n9)(n11)(n8), inner sep=3pt] (m1_3) {};
        \node[draw=red!70!gray, fill=red!30, fill opacity = 0.4, thick, dashed, rounded corners=6pt, fit=(n10)(n12), inner sep=3pt] (m1_4) {};
        \end{pgfonlayer}

        \node[text=red!70!gray, font=\small, align=center] at (0.5, 1.6){\(\leq 2\)};
        \node[text=red!70!gray, font=\small, align=center] at (2, 1.6) {\(\leq 2\)};
        \node[text=red!70!gray, font=\small, align=center] at (3.5, 1.6) {\(\leq 2\)};
        \node[text=red!70!gray, font=\small, align=center] at (5, 1.6) {\(\leq 2\)};
        \node[text=blue!70!gray, font=\small, align=center] at (0.5, -0.6) {\(\leq 1\)};
        \node[text=blue!70!gray, font=\small, align=center] at (2, -0.6) {\(\leq 2\)};
        \node[text=blue!70!gray, font=\small, align=center] at (3.5, -0.6) {\(\leq 1\)};
        \node[text=blue!70!gray, font=\small, align=center] at (5, -0.6) {\(\leq 2\)};

        \end{tikzpicture}
\caption{An example on which the greedy approach does not work. Blue regions indicate the partition classes of $M_1$, red dashed regions indicate the partition classes of $M_2$, and the numbers indicate the upper quotas.}
\label{fig:greedy}
\end{figure}
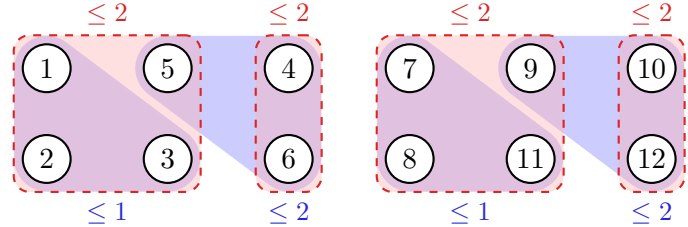

\section{Omitted Proofs from Section~\ref{sec:matroid_flag_matroid} (Quotient Relations)}\label{app:quotient}

In this section, we provide proofs omitted in Section~\ref{sec:matroid_flag_matroid}.

\subsection{Tools Used in the Proofs}
In the proofs of Propositions~\ref{prop:quotient-partition} and \ref{prop:quotient-laminar} (sufficient conditions for generalized partition matroids or generalized laminar matroids to have a quotient relation), we use the following characterization of a matroid pair with a quotient relation.

\begin{theorem}[e.g.,\protect{\cite[Definition~1.1]{brandenburg2024quotients}}]\label{thm:base_quotient}
Let $M$ and $N$ be matroids on the same ground set, and
let $\B(M)$ and $\B(N)$ be their base families, respectively.
Then $N$ is a quotient of $M$ if and only if the following condition holds:

\begin{description}
    \item[{\em (Q)}] For any $X\in \B(N)$, $Y\in \B(M)$, and $x\in X\setminus Y$, there exists $y\in Y\setminus X$ such that
    \[X\setminus\{x\}\cup\{y\}\in \B(N), \quad Y\cup\{x\}\setminus \{y\}\in \B(M).\]
\end{description}
\end{theorem}

In the proof of Proposition~\ref{prop:quotient-transversal} (a sufficient condition for generalized transversal matroids to have a quotient relation), we use an operation of {\em matroid induction}, which is defined as follows.
Given a bipartite graph $G=(U,V; F)$ and a matroid $M_V=(V, \B_V)$, define
\[
\B_U\coloneqq \set{\partial X\cap U \mid X\subseteq F \text{ is a matching and }
\partial X\cap V\in \B_V},
\]
where $\partial X\subseteq U\cup V$ is the set of vertices covered by $X$. If this $\B_U$ is nonempty, then $M_U=(U, \B_U)$ forms a matroid on $U$ (see, e.g., \cite{perfect69, brualdi71}). This operation is called an {\em induction} by the bipartite graph $G$; here $M_U$ is induced from $M_V$.

It was shown by Brandenburg--Loho--Smith~\cite[Lemma~3.12]{brandenburg2024quotients} that the induction operation preserves quotient relations. Although the original statement proves a more general result for M-convex sets and linking systems, we state only the special case needed for our present purpose.
\begin{lemma}[Induction Preserves Quotients \cite{brandenburg2024quotients}]\label{lem:quotient-induction}
Let $G=(U,V;F)$ be a bipartite graph and let $M_V$ and $N_V$ be matroids on $V$ such that $N_V$ is a quotient of $M_V$. Suppose that $M_U$ and $N_U$ are induced via $G$ from $M_V$ and $N_V$, respectively (i.e., the corresponding families are nonempty). Then $N_U$ is a quotient of $M_U$.
\end{lemma}

\subsection{Omitted Proofs}
We provide omitted proofs from Section~\ref{sec:matroid_flag_matroid} in order.

\PropPartitionQuotient*
\begin{proof}
Recall that the base family $\B(N)=\B_{k,\boldsymbol{\ell},\boldsymbol{u}}$ of $N$ is defined as
\[
\B(N)=
\set{B\subseteq E |~|B|=k,\ \ \ell_i \le |B\cap G_i|\le u_i
\ \text{for all } i\in[g]}
\]
and the base family $\B(M)=\B_{k',\boldsymbol{\ell}',\boldsymbol{u}'}$ of $M$ is defined similarly using $(k', \boldsymbol{\ell}',\boldsymbol{u}')$.

We show that they satisfy the axiom (Q) in Theorem~\ref{thm:base_quotient}. Take any $X\in \B(N)$, $Y\in \B(M)$, and $x\in X\setminus Y$. For each $i\in[g]$, we write $X_i=X\cap G_i$ and $Y_i=Y\cap G_i$. Let $i^*\in [g]$ be the index such that $x\in X_{i^*}$. If $Y_{i^*}\setminus X_{i^*}\neq \emptyset$, then any $y\in Y_{i^*}\setminus X_{i^*}$ is the desired element. Indeed, the sets $X'\coloneqq X\setminus \{x\}\cup \{y\}$ and $Y'\coloneqq Y\cup\{x\}\setminus \{y\}$ satisfy $|X'\cap G_i|=|X_i|$ and $|Y'\cap G_i|=|Y_i|$ for every $i\in [g]$, and thus $X'\in \B(N)$ and $Y'\in \B(M)$.

For the rest of the proof, we assume $Y_{i^*}\setminus X_{i^*}=\emptyset$. Since we have $x\in X_{i^*}\setminus Y_{i^*}$, we then have $Y_{i^*}\subsetneq X_{i^*}$, and hence $\ell_{i^*}\leq \ell'_{i^*}\leq |Y_{i^*}|<|X_{i^*}|\leq u_{i^*}\leq u'_{i^*}$, where inequalities other than the center one are due to $X\in \B(N)$, $Y\in\B(M)$ and the conditions on $\boldsymbol{\ell}', \boldsymbol{u}'$.
This sequence of inequalities in particular implies
\begin{equation}\label{eq:i-star}
|Y_{i^*}|<|X_{i^*}|,\quad
\ell_{i^*}<|X_{i^*}|, \quad
|Y_{i^*}|<u'_{i^*}.
\end{equation}
We show that there exists $j\in [g]$ such that
\begin{equation}\label{eq:j}
|X_{j}|<|Y_{j}|,\quad
\ell'_{j}<|Y_{j}|, \quad
|X_{j}|<u_{j}.
\end{equation}
Indeed, if there is $j\in [g]$ satisfying \eqref{eq:j}, then $|X_{j}|<|Y_{j}|$ implies that $Y_j\setminus X_j \neq \emptyset$ and for any $y\in Y_j\setminus X_j$, the second and third inequalities in \eqref{eq:i-star} and \eqref{eq:j} imply that $X'\coloneqq X\setminus \{x\}\cup \{y\}$ and $Y'\coloneqq Y\cup\{x\}\setminus \{y\}$ satisfy the lower and upper quotas of all $i\in[g]$.
Thus, showing the existence of $j\in [g]$ satisfying \eqref{eq:j} completes the proof.

Suppose, to the contrary, that there is no such $j\in [g]$. Then, for any $i\in [g]$, we have at least one of (a) $|X_{i}|\geq |Y_{i}|$, (b)
$\ell'_{i}=|Y_{i}|$, or (c) $|X_{i}|=u_{i}$.

\begin{itemize}
\item In case (a), since $\frac{k'}{k}>1$, we have
$|Y_i|\leq |X_i|\leq \lfloor\tfrac{k'}{k}\cdot |X_i|\rfloor$.

\item In case (b), we have
$|Y_i|=\ell'_i\leq \lfloor\tfrac{k'}{k}\,\ell_i\rfloor\leq \lfloor\tfrac{k'}{k}\cdot |X_i|\rfloor$.

\item In case (c), we have
$|Y_i|\leq u'_i\leq \lfloor\tfrac{k'}{k}\,u_i\rfloor= \lfloor\tfrac{k'}{k}\cdot |X_i|\rfloor$.
\end{itemize}

\noindent
Thus, $|Y_i|\leq  \lfloor\tfrac{k'}{k}\cdot |X_i|\rfloor$ for all $i\in[g]$.
Note that we have $|Y_{i^*}|<|X_{i^*}|$, which implies
$|Y_{i^*}|<  \lfloor\tfrac{k'}{k}\cdot |X_{i^*}|\rfloor$.
Thus, by summing these inequalities over all $i\in[g]$, we obtain
\[k'=|Y|=\sum_{i\in [g]}|Y_i|<\sum_{i\in[g]}\lfloor\tfrac{k'}{k}\cdot |X_{i}|\rfloor\leq \tfrac{k'}{k}\sum_{i\in [g]}|X_i|=\tfrac{k'}{k}\cdot |X|=k',
\]
a contradiction.
\end{proof}

\begin{corollary}\label{cor:block-flag}
Let $E=G_1\sqcup \cdots \sqcup G_g$, $\boldsymbol{\alpha}, \boldsymbol{\beta}\in [0,1]^g$, and $k,b\in[n]$ be given as in Definition~\ref{def:block-fair}.
Assume that an $(\boldsymbol{\alpha},\boldsymbol{\beta})$-block-$k$-fair ranking exists.
Then there exists a flag matroid $\bM=(M_1, \dots, M_s)$ of rank $(h_1, \dots, h_s)$ such that, for any ranking $\pi:[n]\to E$, we have that $\pi$ is $(\boldsymbol{\alpha},\boldsymbol{\beta})$-block-$k$-fair if and only if $(\Top_{h_1}(\pi), \dots, \Top_{h_s}(\pi))\in \F(\bM)$.
\end{corollary}
\begin{proof}
If no $h\in[n]$ satisfies $h\ge k$ and $h\equiv 0\pmod b$, then every ranking is block-$k$-fair, and the claim holds by taking $s=1$, $h_1=n$, and $M_1$ to be the free matroid on $E$.
Otherwise, let $\{h_1, \dots, h_s\}=\set{h\in[n]\mid h\ge k,\ h\equiv 0 \pmod b}$, where $h_1<h_2<\cdots <h_s$, and for each $j=1,\dots,s$ define a matroid $M_j=(E, \B_j)$ by letting
$\B_j$ be the family of sets $B\subseteq E$ with $|B|=h_j$ and $\alpha_i h_j\leq |B\cap G_i|\leq \beta_i h_j$ for all $i\in[g]$.
Each $\B_j$ is nonempty by the assumed existence of a block-$k$-fair ranking.

For any $j\in[s-1]$, we have $h_j=bd$ and $h_{j+1}=b(d+1)$ for some integer $d$.
Since $\alpha_i b$ is an integer (as assumed in Definition~\ref{def:block-fair}), both $\alpha_i h_j$ and $\alpha_i h_{j+1}$ are integers, and they satisfy
$\alpha_i h_{j+1}=\frac{h_{j+1}}{h_j}\,\alpha_i h_j$.
The same holds for $\beta_i$.
Hence, the sufficient condition in Proposition~\ref{prop:quotient-partition} is satisfied.
Thus, $M_j$ is a quotient of $M_{j+1}$ for all $j$, and hence $(M_1,\dots,M_s)$ forms a flag matroid.
Then, the claim is clear from the definition of block-$k$-fairness.
\end{proof}

\PropTransversalQuotient*
\begin{proof}
Let $\widehat{E}\coloneqq E\times D$ and
$\widehat{E}_d\coloneqq \set{(e,d)\in \widehat{E} \mid e\in E}$ for each $d\in D$. Then, $\{\widehat{E}_d\}_{d\in D}$ forms a partition of $\widehat{E}$. Using this partition and the given parameters, define two generalized partition matroids $N_{\rm part}=(\widehat{E},\B_{k,\boldsymbol{p},\boldsymbol{q}})$ and $M_{\rm part}=(\widehat{E},\B_{k',\boldsymbol{p}',\boldsymbol{q}'})$ as in Example~\ref{ex:g-partition}.
By Proposition~\ref{prop:quotient-partition}, $N_{\rm part}$ is a quotient of $M_{\rm part}$.

Next, define a bipartite graph $G=(E,\widehat{E};F)$ by
$F\coloneqq \set{(e,(e,d))\in E\times \widehat{E} \mid e\in \Gamma(d)}$.
Let $N=(E, \B)$ and $M=(E, \B')$ be the matroids induced via $G$ from $N_{\rm part}$ and $M_{\rm part}$, respectively. Then, by Lemma~\ref{lem:quotient-induction}, $N$ is a quotient of $M$.

It remains to show that $\B=\B^{\rm tv}_{k,\boldsymbol{p},\boldsymbol{q}}$ and $\B'=\B^{\rm tv}_{k',\boldsymbol{p}',\boldsymbol{q}'}$. We prove only the former, since the latter is identical. If $B\in \B^{\rm tv}_{k,\boldsymbol{p},\boldsymbol{q}}$, then there is a map $\varphi:B\to D$ such that $\varphi^{-1}(d)\subseteq \Gamma(d)$ and $p_d\le |\varphi^{-1}(d)|\le q_d$ for all $d\in D$. Then $X_\varphi\coloneqq \{(e,(e,\varphi(e))) \mid e\in B\}$ is a matching in $G$, and we have $\partial X_\varphi\cap E=B$ and $\partial X_\varphi\cap \widehat{E}\in \B_{k,\boldsymbol{p},\boldsymbol{q}}$. Hence $B\in \B$.

Conversely, if $B\in \B$, then there is a matching $X_B\subseteq F$ such that $\partial X_B\cap E=B$ and $\partial X_B\cap \widehat{E}\in \B_{k,\boldsymbol{p},\boldsymbol{q}}$. For each $e\in B$, since $X_B$ is a matching, there is a unique $d\in D$ with $(e,(e,d))\in X_B$; define $\varphi_B(e)\coloneqq d$. Then $(e,(e,d))\in F$ implies $e\in \Gamma(d)$. Moreover, since $\partial X_B\cap \widehat{E}\in \B_{k,\boldsymbol{p},\boldsymbol{q}}$, we have $|B|=k$ and $p_d\le |\varphi_B^{-1}(d)|\le q_d$ for all $d\in D$. Hence $B\in \B^{\rm tv}_{k,\boldsymbol{p},\boldsymbol{q}}$.

Therefore $\B=\B^{\rm tv}_{k,\boldsymbol{p},\boldsymbol{q}}$. The proof of $\B'=\B^{\rm tv}_{k',\boldsymbol{p}',\boldsymbol{q}'}$ is identical.
\end{proof}

\PropLaminarQuotient*
\begin{proof}
We first prepare some terminology and assumptions on the laminar family $\L\subseteq 2^E$.
For two members $L, L'\in \L$, we say that $L'$ is a {\em child} of $L$ (and also $L$ is a {\em parent} of $L'$) if $L'\subsetneq L$ and there is no $L''\in\L$ with $L'\subsetneq L''\subsetneq L$. For any $L\in \L$, we write $\ch(L)$ for the set of children of $L$. A member $L\in \L$ with $\ch(L)=\emptyset$ is called a {\em leaf}.
We assume the following on the laminar family $\L$.
\begin{enumerate}
    \item $E\in \L$.
    \item If $L\in \L$ is not a leaf, then $L$ is partitioned into its children.
\end{enumerate}

Indeed, we can impose these assumptions without loss of generality. If the original $\L$ does not satisfy the first condition, then we can add $E$ to $\L$ and set $g(E)=f(E)=k$ and $g'(E)=f'(E)=k'$. If the second condition is not satisfied for some $L\in \L$, then add $L^*\coloneqq L\setminus \cup_{L'\in \ch(L)} L'$ to $\L$ and set $g(L^*)=0$, $f(L^*)=|L^*|$, $g'(L^*)=0$, and $f'(L^*)=\lfloor \tfrac{k'}{k}|L^*|\rfloor$. Each of these operations does not change the base families $\B^{\rm lam}_{k,g,f}$ and $\B^{\rm lam}_{k',g',f'}$ and preserves the condition in Proposition~\ref{prop:quotient-laminar}.

As we have $E\in \L$, the laminar family $\L$ can be seen as a tree rooted at $E$. For each $L\in \L$, its {\em ancestor path} is a path from $L$ to $E$ on this tree.

We show that $\B(N)=\B^{\rm lam}_{k,g,f}$ and $\B(M)=\B^{\rm lam}_{k',g',f'}$ satisfy (Q) in Theorem~\ref{thm:base_quotient}.
Take any $X\in \B(N)$, $Y\in \B(M)$, and $x\in X\setminus Y$. Let $L^x\in \L$ be the unique leaf containing $x$. If $(Y\setminus X)\cap L^x\neq \emptyset$, then we see that any $y\in (Y\setminus X)\cap L^x$ satisfies the desired conditions in (Q) and we are done. Thus, in the rest, assume that $(Y\setminus X)\cap L^x=\emptyset$.

As we have $x\in (X\setminus Y)\cap L^x$, we obtain $Y\cap L^x\subsetneq X\cap L^x$. In particular, we have $|Y\cap L^x|<|X\cap L^x|\le \lfloor\tfrac{k'}{k}|X\cap L^x|\rfloor$.
Along the ancestor path of $L^x$, take the ancestor $L^*$ closest to $L^x$ subject to
\[|Y\cap L^*|\geq \lfloor\tfrac{k'}{k}|X\cap L^*|\rfloor.\]
Such an ancestor must exist because the root ancestor $E$ satisfies $|Y\cap E|=k'=\lfloor\tfrac{k'}{k}k\rfloor=\lfloor\tfrac{k'}{k}|X\cap E|\rfloor$. In addition, $L^*$ must not be a leaf by definition. From this $L^*$ we go down the tree toward a leaf using the following claims.

\begin{claim}\label{claim:first-step}
There exists a child $L$ of $L^*$ that satisfies $|Y\cap L|> \lfloor\tfrac{k'}{k}|X\cap L|\rfloor$.
\end{claim}
\begin{proof}
Suppose, to the contrary, there is no child satisfying the required condition. Then, every $L\in \ch(L^*)$ satisfies $|Y\cap L|\leq \lfloor\tfrac{k'}{k}|X\cap L|\rfloor$. By the choice of $L^*$, the child of $L^*$ on the ancestor path from $L^x$ to $L^*$, denoted by $L'$, satisfies $|Y\cap L'|< \lfloor\tfrac{k'}{k}|X\cap L'|\rfloor$. Since $L^*$ is partitioned into $\ch(L^*)$, we have
\[|Y\cap L^*|~=\sum_{L\in \ch(L^*)} |Y\cap L|~~<\sum_{L\in \ch(L^*)} \lfloor\tfrac{k'}{k}|X\cap L|\rfloor~\leq~ \lfloor\tfrac{k'}{k}|X\cap L^*|\rfloor,\]
while we have $|Y\cap L^*|\geq \lfloor\tfrac{k'}{k}|X\cap L^*|\rfloor$, a contradiction.
\end{proof}

\begin{claim}\label{claim:next-step}
If $L\in \L$ satisfies $|Y\cap L|> \lfloor\tfrac{k'}{k}|X\cap L|\rfloor$ and is not a leaf, then at least one of its children $L'\in \ch(L)$ also satisfies $|Y\cap L'|> \lfloor\tfrac{k'}{k}|X\cap L'|\rfloor$.
\end{claim}
\begin{proof}
Suppose, to the contrary, no child of $L$ satisfies the required condition. Then, every $L'\in \ch(L)$ satisfies $|Y\cap L'|\leq \lfloor\tfrac{k'}{k}|X\cap L'|\rfloor$, but then summing this inequality implies
\[|Y\cap L|=\sum_{L'\in \ch(L)} |Y\cap L'|\leq \sum_{L'\in \ch(L)} \lfloor\tfrac{k'}{k}|X\cap L'|\rfloor \leq \lfloor\tfrac{k'}{k}|X\cap L|\rfloor,\]
contradicting the assumption of the claim.
\end{proof}

By combining Claims~\ref{claim:first-step} and \ref{claim:next-step}, we can find a path from $L^*$ to a leaf such that every member $L$ on this path except $L^*$ satisfies $|Y\cap L|>\lfloor\tfrac{k'}{k}|X\cap L|\rfloor$. Let $L^y$ be this terminal leaf. To summarize, we have found $L^*$ and $L^y$ such that

\begin{enumerate}
\item[~~~\sf (a)] $L^*$ is a common ancestor of $L^x$ and $L^y$,
\item[~~~\sf (b)] $|Y\cap L|<\lfloor\tfrac{k'}{k}|X\cap L|\rfloor$ for every $L$ on the $L^x$--$L^*$ ancestor path excluding $L^*$, and
\item[~~~\sf (c)] $|Y\cap L|>\lfloor\tfrac{k'}{k}|X\cap L|\rfloor$ for every $L$ on the $L^y$--$L^*$ ancestor path excluding $L^*$.
\end{enumerate}

The third condition in particular implies $|Y\cap L^y|>\lfloor\tfrac{k'}{k}|X\cap L^y|\rfloor\geq |X\cap L^y|$, and hence $(Y\setminus X)\cap L^y\neq \emptyset$. Take any $y\in (Y\setminus X)\cap L^y$. We show that this is our desired $y$, i.e.,
\[X'\coloneqq X\setminus\{x\}\cup\{y\}\in \B(N), \quad Y'\coloneqq Y\cup\{x\}\setminus \{y\}\in \B(M).\]
To this end, we need to show that $g(L)\leq |X'\cap L|\leq f(L)$ and $g'(L) \leq |Y'\cap L|\leq f'(L)$ for every $L\in \L$.
By the definitions of $X'$ and $Y'$, the following hold.

\begin{itemize}
\item[~~~\sf (d)] $|X'\cap L|=|X\cap L|-1$ and $|Y'\cap L|=|Y\cap L|+1$ for every $L$ on the $L^x$--$L^*$ ancestor path excluding $L^*$,
\item[~~~\sf (e)] $|X'\cap L|=|X\cap L|+1$ and $|Y'\cap L|=|Y\cap L|-1$ for every $L$ on the $L^y$--$L^*$ ancestor path excluding $L^*$, and
\item[~~~\sf (f)] $|X'\cap L|=|X\cap L|$ and $|Y'\cap L|=|Y\cap L|$ for all other $L\in\L$.
\end{itemize}

As we have conditions (b) and (c) listed above, to prove $g(L)\leq |X'\cap L|\leq f(L)$ and $g'(L) \leq |Y'\cap L|\leq f'(L)$ for every $L\in \L$, it suffices to show the following two claims hold for any $L\in \L$.
\begin{claim}\label{claim:laminar3}
If $|Y\cap L|<\lfloor\tfrac{k'}{k}|X\cap L|\rfloor$, then $g(L)< |X\cap L|$ and $|Y\cap L|< f'(L)$.
\end{claim}
\begin{proof}
We prove the contrapositive. Suppose that we have $g(L)=|X\cap L|$ or $|Y\cap L|= f'(L)$. In the former case, we have $|Y\cap L|\geq g'(L)=\lfloor\tfrac{k'}{k} g(L)\rfloor=\lfloor\tfrac{k'}{k} |X\cap L|\rfloor$. In the latter case, we have $|Y\cap L|=f'(L)=\lfloor\tfrac{k'}{k} f(L)\rfloor\geq \lfloor\tfrac{k'}{k} |X\cap L|\rfloor$. Thus, in both cases, we obtain $|Y\cap L|\geq \lfloor\tfrac{k'}{k}|X\cap L|\rfloor$.
\end{proof}

\begin{claim}\label{claim:laminar4}
If $|Y\cap L|>\lfloor\tfrac{k'}{k}|X\cap L|\rfloor$, then $|X\cap L|<f(L)$ and $g'(L)<|Y\cap L|$.
\end{claim}
\begin{proof}
The proof is similar to that of the above claim. Suppose that $|X\cap L|=f(L)$ or $g'(L)=|Y\cap L|$. In the former case, we have $|Y\cap L|\leq f'(L)=\lfloor\tfrac{k'}{k} f(L)\rfloor=\lfloor\tfrac{k'}{k} |X\cap L|\rfloor$. In the latter case, we have $|Y\cap L|=g'(L)=\lfloor\tfrac{k'}{k} g(L)\rfloor\leq \lfloor\tfrac{k'}{k} |X\cap L|\rfloor$. Thus, in both cases, $|Y\cap L|\leq \lfloor\tfrac{k'}{k}|X\cap L|\rfloor$.
\end{proof}

By combining (b), (d), and Claim~\ref{claim:laminar3}, we see that $g(L)\leq |X'\cap L|\leq f(L)$ and $g'(L) \leq |Y'\cap L|\leq f'(L)$ for every $L\in \L$ on the $L^x$--$L^*$ ancestor path excluding $L^*$.
By combining (c), (e), and Claim~\ref{claim:laminar4}, the same inequalities hold for members on the $L^y$--$L^*$ ancestor path excluding $L^*$. For all other members of $\L$, (f) guarantees that the required inequalities are preserved. Thus, we have shown that $X'\in \B(N)$ and $Y'\in \B(M)$.
\end{proof}

\section{Omitted Proofs from Section~\ref{sec:matroid-CFR} (Matroid-CFR)}\label{app:CFR}

Here we provide a few proofs omitted from Section~\ref{sec:matroid-CFR}.

\ObsOptRanking*
\begin{proof}
Let $P_i=B_i\setminus B_{i-1}$ for $i\in[s+1]$, where we set for convenience $B_0=\emptyset$ and $B_{s+1}=E$. Then any ranking $\sigma$ satisfying $\Top_{k_i}(\sigma)=B_i~(i\in[s])$ must place the blocks $P_1,P_2,\dots,P_{s+1}$ in this order.
Now, the pairs of elements counted in $d_{\mathrm{KT}}(\sigma,\pi)$ are partitioned into two types: pairs belonging to different blocks, and pairs belonging to the same block. For pairs of elements belonging to different blocks, their relative order in $\sigma$ is completely determined by the fixed block order $P_1,P_2,\dots,P_{s+1}$. Hence the total contribution of such pairs to $d_{\mathrm{KT}}(\sigma,\pi)$ is the same for all rankings $\sigma$ satisfying $\Top_{k_i}(\sigma)=B_i~(i\in[s])$.

Therefore, to minimize $d_{\mathrm{KT}}(\sigma,\pi)$, it suffices to minimize the contribution of pairs belonging to the same block. For each block $P_i$, this contribution is exactly the number of inversions of the order induced by $\sigma$ on $P_i$ relative to the order induced by $\pi$ on $P_i$. This number is minimized, and equal to $0$, when the elements of $P_i$ are ordered according to $\pi$. Thus, the ranking $\pi_{\bB}$ is optimal.
\end{proof}

\LemOptRanking*
\begin{proof}
By the construction of $\pi_B$, the value $\dKT(\pi_{B}, \pi)$ coincides with the number of pairs $(e,f)\in B\times (E\setminus B)$ with $f \prec_{\pi} e$. Therefore,
\begin{align*}
    \dKT(\pi_B, \pi) &= \sum_{e \in B} ~\big|\set{ f \in E\setminus B | f \prec_{\pi} e }\big| \\
    &= \sum_{e \in B} \big( \pi^{-1}(e) - 1 - |\set{ f \in B | f \prec_{\pi} e }| \,\big) \\
    &= \sum_{e \in B} \pi^{-1}(e) - k - \sum_{e \in B} \big|\set{f \in B |f \prec_{\pi} e}\big|\\
    &= \sum_{e \in B} \pi^{-1}(e) -k -\sum_{j=1}^{k}(j-1)
    = \sum_{e \in B} \pi^{-1}(e) -\frac{k(k+1)}{2}.
\end{align*}
Thus, the required equality is obtained.
\end{proof}

\section{Omitted Proofs from Section~\ref{sec:approx} (Approximation for FRA)}\label{app:approx}

In Section~\ref{sec:approx}, we provide only the statements of our results on approximation algorithms for Matroid-FRA and Flag-Matroid-FRA. Here we provide their details.

\subsection{Matroid-FRA Approximation}
This part is devoted to the proof of the following theorem on Matroid-FRA (Matroid Feasible Rank Aggregation). See Section~\ref{sec:approx} for the definition of the problem.
\ThmApproxMat*

For FRA with $k$-fairness criteria (i.e., a special case of Matroid-FRA with a generalized partition matroid), Chakraborty et al.~\cite{chak25improve} show that there exists a $(2+\varepsilon)$-approximation algorithm for any constant $\varepsilon>0$. Their algorithm consists of the following steps:

\begin{enumerate}
    \item Find an appropriate bi-partition $\{B, E\setminus B\}$ such that $B$ has size $k$ and satisfies $k$-fairness by solving the {\em colorful bi-partition problem} (described below in a generalized form).
    \item Apply the known PTAS for rank aggregation (without constraints) by Kenyon-Mathieu and Schudy~\cite{KenyonMathieuSchudy07} to each of $(\pi_i|_{B})_{i\in [m]}$ and $(\pi_i|_{E\setminus B})_{i\in [m]}$ to obtain $\sigma_B$ and $\sigma_{E\setminus B}$.
    \item Output a concatenation of $\sigma_B$ and $\sigma_{E\setminus B}$ as $\sigma$.
\end{enumerate}

To generalize their algorithm to Matroid-FRA, we consider solving the following generalized version of the colorful bi-partition subproblem. As in the original version, we describe the instance by a \emph{weighted tournament}. 
In this paper, a weighted tournament is a directed graph $T=(E,A)$ in which, for every pair of distinct vertices $e,f\in E$, both arcs $(e,f)$ and $(f,e)$ belong to $A$, together with a weight function $c:A\to\mathbb{R}$.\footnote{This terminology is used in some works on weighted rank aggregation. In graph theory, a tournament usually means an orientation of a complete graph.}
A weighted tournament arises from a collection of rankings as follows. Consider a weighted tournament with $E$ being the set of candidates and, for each $e,f\in E$, set the weight $c(e,f)$ to be the fraction of voters preferring $e$ to $f$ among all voters. Then we see that the weights satisfy the following properties:

\begin{itemize}
    \item Probability Constraints: $c(e,f)+c(f,e)=1$ for any distinct $e,f\in E$, and
    \item Triangle Inequality: $c(e,g)\leq c(e,f)+c(f,g)$ for any distinct $e,f,g\in E$.
\end{itemize}

The following problem is a matroidal generalization of the colorful bi-partition problem in \cite{chak25improve}. For a weighted tournament $T=(E,A)$ and any vertex subset $S\subseteq E$, we define $\cost(S)=\sum_{(e,s)\in A:\, s\in S,\, e\in E\setminus S}c(e,s)$. Also, for any vertex $e\in E$, we define $\delta^{-}(e)=\sum_{f\in E\setminus\{e\}} c(f,e)$.

\begin{tcolorbox}[
  colback=white,
  colframe=black!75,
  boxrule=0.25pt,
  arc=8pt
]
\textsc{Matroid Bi-partition}\\[2pt]
\textbf{Input:} A weighted tournament $(T=(E,A),c)$ and a matroid $M=(E,\B)$ of rank $k$.\\
\textbf{Output:} A base $B\in\B$ minimizing $\cost(B)=\sum_{(e,b)\in A:\, b\in B,\, e\in E\setminus B} c(e,b)$.
\end{tcolorbox}

\begin{lemma}
When the weights $c$ satisfy the probability constraints, for any base $B\in \B$, we have
$\cost(B)=\sum_{b\in B}\delta^{-}(b)-\frac{k(k-1)}{2}$.
\end{lemma}
\begin{proof}
For any base $B\in \B$, we have
\[
\sum_{b \in B}\delta^{-}(b)
=
\sum_{b \in B} \sum_{e \in E\setminus\{b\}} c(e,b)
=
\sum_{b\in B}\sum_{b' \in B\setminus\{b\}} c(b',b)
+
\sum_{b \in B,\, e \in E\setminus B} c(e,b).
\]
The second term is $\cost(B)$. By the probability constraints, the first term is $\frac{k(k-1)}{2}$, since for each unordered pair $\{b,b'\}\subseteq B$ we have $c(b,b')+c(b',b)=1$. Thus, the lemma follows.
\end{proof}

By this lemma, we see that Matroid Bi-partition reduces to the minimum-weight base problem for matroids, which can be solved by the greedy algorithm with $O(|E|)$ calls to the independence oracle.

\begin{lemma}
Matroid Bi-partition can be solved in polynomial time when the weights satisfy the probability constraints.
\end{lemma}
\begin{proof}
Define a weight function $w:E\to\mathbb{R}$ by $w(e)=\delta^{-}(e)$.
By the previous lemma, for every base $B\in\B$,
\[
\cost(B)=\sum_{b\in B}w(b)-\frac{k(k-1)}{2}.
\]
Hence minimizing $\cost(B)$ over $B\in\B$ is equivalent to finding a minimum-weight base of the matroid $M$, which can be done by the greedy algorithm in polynomial time.
\end{proof}

We modify the algorithm by Chakraborty et al.~\cite{chak25improve} by replacing the subroutine solving colorful bi-partition with the one solving Matroid Bi-partition.
Specifically, we construct a ranking $\sigma^*$ in the following manner:

\begin{itemize}
\item Given an instance $I=(\pi_1,\dots, \pi_m, M=(E,\B))$ of Matroid-FRA, define a weighted tournament $(T=(E, A), c)$ by setting $c(e,f)=\frac{1}{m}|\set{j\in [m]\mid e\prec_{\pi_j} f}|$.
\item Let $B^*$ be the optimal solution to Matroid Bi-partition for the input $(T, c)$ and $M$.
\item Apply the known PTAS for rank aggregation (without constraints) by Kenyon-Mathieu and Schudy~\cite{KenyonMathieuSchudy07} to each of $(\pi_i|_{B^*})_{i\in [m]}$ and $(\pi_i|_{E\setminus B^*})_{i\in [m]}$ to obtain $\sigma_{B^*}$ and $\sigma_{E\setminus B^*}$.
\item Output a concatenation of $\sigma_{B^*}$ and $\sigma_{E\setminus B^*}$ as $\sigma^*$.
\end{itemize}

We reproduce the analysis of the approximation ratio in~\cite{chak25improve} for completeness.

\begin{lemma}
The ranking $\sigma^*$ constructed as above is a $(2+\varepsilon)$-approximate solution to Matroid-FRA.
\end{lemma}
\begin{proof}
For any subset $S\subseteq E$ and any ranking $\tau$ over $S$, define
\[
\mathrm{RA}(S,\tau)\coloneqq \sum_{i=1}^m d_{\mathrm{KT}}(\tau,\pi_i|_{S}),
\]
and let
$\mathrm{OPT}(S)\coloneqq \min\set{\mathrm{RA}(S,\tau)| \tau \text{ is a ranking over } S}$.
Also, let $\mathrm{OPT}$ denote the optimal objective value of the given instance of Matroid-FRA, and let $\hat{\sigma}$ be an optimal feasible ranking. Put $\hat{B}\coloneqq \Top_k(\hat{\sigma})$.

Fix any base $B\in \B$, and let $\sigma$ be any ranking with $\Top_k(\sigma)=B$. Then the unordered pairs of elements of $E$ are partitioned into three types: pairs contained in $B$, pairs contained in $E\setminus B$, and pairs with one endpoint in each of $B$ and $E\setminus B$. Hence
\begin{equation}\label{eq:matroid-fra-decomp}
\sum_{i=1}^m d_{\mathrm{KT}}(\sigma,\pi_i)
=
m\cdot \cost(B)
+
\mathrm{RA}(B,\sigma|_B)
+
\mathrm{RA}(E\setminus B,\sigma|_{E\setminus B}).
\end{equation}
Indeed, for every pair $(b,e)\in B\times (E\setminus B)$, the ranking $\sigma$ places $b$ before $e$, so the contribution of this pair to the total Kendall tau distance is exactly the number of voters $j\in[m]$ with $e\prec_{\pi_j} b$, i.e., $m\cdot c(e,b)$. Summing over all such pairs gives the term $m\cdot\cost(B)$.

Applying \eqref{eq:matroid-fra-decomp} to the optimal ranking $\hat{\sigma}$ and its top-$k$ set $\hat{B}$, we obtain $m\cdot \cost(\hat{B})\leq \mathrm{OPT}$.
Since $B^*$ is an optimal solution to Matroid Bi-partition, it follows that
\begin{equation}\label{eq:cross-bound}
m\cdot \cost(B^*)\leq m\cdot \cost(\hat{B})\leq \mathrm{OPT}.
\end{equation}
By the guarantee of the PTAS for unconstrained rank aggregation,
\[
\mathrm{RA}(B^*,\sigma_{B^*})\leq (1+\varepsilon)\mathrm{OPT}(B^*),
\qquad
\mathrm{RA}(E\setminus B^*,\sigma_{E\setminus B^*})\leq (1+\varepsilon)\mathrm{OPT}(E\setminus B^*).
\]
On the other hand, the restrictions $\hat{\sigma}|_{B^*}$ and $\hat{\sigma}|_{E\setminus B^*}$ are rankings over $B^*$ and $E\setminus B^*$, respectively. Therefore,
\[
\mathrm{OPT}(B^*)+\mathrm{OPT}(E\setminus B^*)
\leq
\mathrm{RA}(B^*,\hat{\sigma}|_{B^*})+\mathrm{RA}(E\setminus B^*,\hat{\sigma}|_{E\setminus B^*}).
\]
The right-hand side counts only the contributions of pairs contained entirely in $B^*$ or entirely in $E\setminus B^*$, and hence it is at most the whole objective value of $\hat{\sigma}$. Thus,
\begin{equation}\label{eq:inside-bound}
\mathrm{OPT}(B^*)+\mathrm{OPT}(E\setminus B^*)\leq \mathrm{OPT}.
\end{equation}

Finally, applying \eqref{eq:matroid-fra-decomp} to $\sigma^*$ and $B^*$, and using \eqref{eq:cross-bound} and \eqref{eq:inside-bound}, we get
\begin{align*}\textstyle
\sum_{i=1}^m d_{\mathrm{KT}}(\sigma^*,\pi_i)
&=
m\cdot \cost(B^*)
+
\mathrm{RA}(B^*,\sigma_{B^*})
+
\mathrm{RA}(E\setminus B^*,\sigma_{E\setminus B^*})\\
&\leq
m\cdot \cost(B^*)
+
(1+\varepsilon)\bigl(\mathrm{OPT}(B^*)+\mathrm{OPT}(E\setminus B^*)\bigr)\\
&\leq
\mathrm{OPT}+(1+\varepsilon)\mathrm{OPT}\\
&=
(2+\varepsilon)\mathrm{OPT}.
\end{align*}
Therefore, $\sigma^*$ is a $(2+\varepsilon)$-approximate solution to Matroid-FRA.
\end{proof}

\subsection{Flag-Matroid-FRA Approximation}

In the previous subsection, we considered a matroidal generalization of FRA with $k$-fairness, i.e., only one prefix is required to satisfy the constraint.
In this subsection, we consider the flag matroid analogue, which includes FRA with block-$k$-fairness as a special case.

We apply the generic framework by Chakraborty et al.~\cite{chak25improve}.
They proved that, if there is an efficient procedure to solve the closest feasible ranking problem with some constraint, then we can obtain an efficient $2.881$-approximation algorithm for the rank aggregation problem with the same constraint. Since we have shown that Flag-Matroid-CFR is solvable in polynomial time, the framework applies directly to the flag matroid constraint, and we obtain the following.

\ThmApproxFlag*
\begin{proof}
The claim follows immediately from the generic $2.881$-approximation framework of Chakraborty et al.~\cite{chak25improve} together with the polynomial-time solvability of Flag-Matroid-CFR shown in Corollary~\ref{cor:flag-CFR}.
\end{proof}

\section{Omitted Proofs from Section~\ref{sec:hardness} (Hardness with Few Voters)}\label{app:hardness}

We provide proofs of the hardness results for Matroid-FRA with a constant number of voters. Specifically, we show NP-hardness for the cases of $2$ voters and $3$ voters, and then deduce hardness for all larger constant values of $m$.

In the proofs, we use the terminology that $\pi$ and $\sigma$ are {\em inconsistent} on an unordered pair $\{e,f\}\subseteq E$ if either $e\prec_\pi f$ and $e\succ_\sigma f$, or $e\succ_\pi f$ and $e\prec_\sigma f$.

\ThmHardTwo*

To obtain the theorem, we show NP-completeness of the following decision version.
\begin{theorem}
It is NP-complete to determine, given a number $\theta$ and an instance $I$ of Matroid-FRA with $m=2$ and a partition matroid specified by its partition and capacities, whether $I$ admits a feasible ranking with objective value at most $\theta$.
\end{theorem}

\begin{proof}
It is clear that the problem is in NP. We show NP-hardness by a reduction from \emph{Perfect $3$-Dimensional Matching} ($3$DM), a classical NP-complete problem~\cite{karp1972,garey1979}.
An instance of $3$DM consists of three pairwise disjoint sets
$X$, $Y$, and $Z$ with $|X| = |Y| = |Z| = q$,
and a set $T \subseteq X \times Y \times Z$ of triples.
The question is whether there exists a perfect $3$-dimensional matching, i.e., a subset $N \subseteq T$ of size $q$ such that every element of $X \cup Y \cup Z$ appears in exactly one triple of $N$.
Without loss of generality, we may assume that every element of $X\cup Y\cup Z$ appears in at least one triple of $T$, since otherwise the instance is trivially a no-instance.

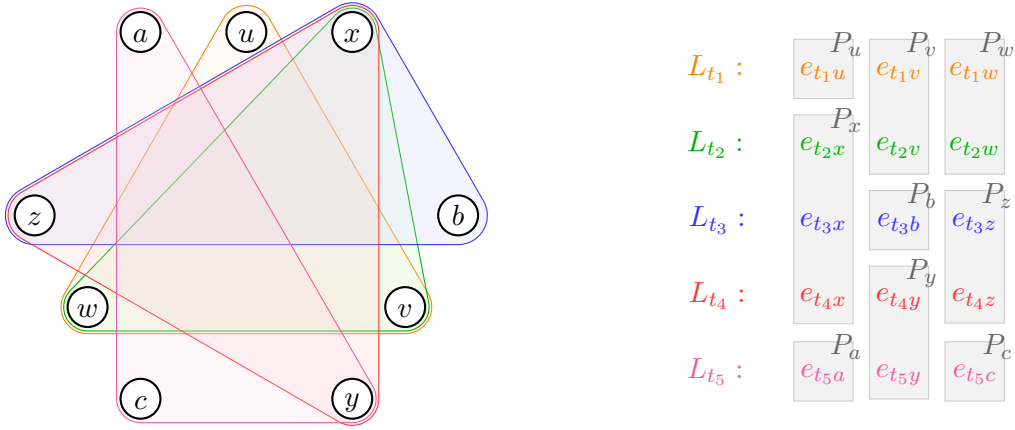
\begin{figure}[t]
    \centering
        \begin{minipage}{0.4\columnwidth}
            \centering
            \begin{tikzpicture}[x=1cm, y=1cm, element/.style={circle, draw=black, thick, fill=white, minimum size=15pt, inner sep=0pt, font=\bfseries}]
            \node[element] (a) at (-1.4, {1.4*sqrt(3)}) {\(\vphantom{A}a\)};
            \node[element] (u) at (0, {1.4*sqrt(3)}) {\(\vphantom{A}u\)};
            \node[element] (x) at (1.4, {1.4*sqrt(3)}) {\(\vphantom{A}x\)};
            \node[element] (b) at (2.8, 0) {\(\vphantom{A}b\)};
            \node[element] (v) at (2.1, -{1.4*sqrt(3)}/2) {\(\vphantom{A}v\)};
            \node[element] (y) at (1.4, -{1.4*sqrt(3)}) {\(\vphantom{A}y\)};
            \node[element] (c) at (-1.4, -{1.4*sqrt(3)}) {\(\vphantom{A}c\)};
            \node[element] (w) at (-2.1, -{1.4*sqrt(3)}/2) {\(\vphantom{A}w\)};
            \node[element] (z) at (-2.8, 0) {\(\vphantom{A}z\)};

            \begin{pgfonlayer}{background}
            \draw[draw = orange, fill = orange!20, fill opacity=0.2] \convexpath{u,v,w}{10pt};
            \draw[draw = green!70!black, fill = green!20, fill opacity=0.2] \convexpath{x,v,w}{9pt};
            \draw[draw = blue!80, fill = blue!20, fill opacity=0.2] \convexpath{x,b,z}{11pt};
            \draw[draw = red!80, fill = red!20, fill opacity=0.2] \convexpath{x,y,z}{10pt};
            \draw[draw = magenta!80, fill = magenta!20, fill opacity=0.2] \convexpath{a,y,c}{9pt};

            \end{pgfonlayer}
            \end{tikzpicture}
        \end{minipage}
    \qquad \qquad
        \begin{minipage}{0.4\columnwidth}
        \centering
        \begin{tikzpicture}[x=1cm, y=1cm, element/.style={circle, draw=none, minimum size=15pt, inner sep=0pt, font=\bfseries}]
        \node[element, text = orange] (1u) at (-1, 2) {\(\vphantom{Ae_y}\)\makebox[\widthof{\(e_{t_1w}\)}]{\(e_{t_1u}\)}};
        \node[element, text = orange] (1v) at (0, 2) {\(\vphantom{Ae_y}\)\makebox[\widthof{\(e_{t_1w}\)}]{\(e_{t_1v}\)}};
        \node[element, text = orange] (1w) at (1, 2) {\(\vphantom{Ae_y}\)\makebox[\widthof{\(e_{t_1w}\)}]{\(e_{t_1w}\)}};
        \node[element, text = green!70!black] (2x) at (-1, 1) {\(\vphantom{Ae_y}\)\makebox[\widthof{\(e_{t_1w}\)}]{\(e_{t_2x}\)}};
        \node[element, text = green!70!black] (2v) at (0, 1) {\(\vphantom{Ae_y}\)\makebox[\widthof{\(e_{t_1w}\)}]{\(e_{t_2v}\)}};
        \node[element, text = green!70!black] (2w) at (1, 1) {\(\vphantom{Ae_y}e_{t_2w}\)};
        \node[element, text = blue!80] (3x) at (-1, 0) {\(\vphantom{Ae_y}e_{t_3x}\)};
        \node[element, text = blue!80] (3b) at (0, 0) {\(\vphantom{Ae_y}\)\makebox[\widthof{\(e_{t_1w}\)}]{\(e_{t_3b}\)}};
        \node[element, text = blue!80] (3z) at (1, 0) {\(\vphantom{Ae_y}\)\makebox[\widthof{\(e_{t_1w}\)}]{\(e_{t_3z}\)}};
        \node[element, text = red!80] (4x) at (-1, -1) {\(\vphantom{Ae_y}e_{t_4x}\)};
        \node[element, text = red!80] (4y) at (0, -1) {\(\vphantom{Ae_y}\)\makebox[\widthof{\(e_{t_1w}\)}]{\(e_{t_4y}\)}};
        \node[element, text = red!80] (4z) at (1, -1) {\(\vphantom{Ae_y}e_{t_4z}\)};
        \node[element, text = magenta!80] (5a) at (-1, -2) {\(\vphantom{Ae_y}\)\makebox[\widthof{\(e_{t_1w}\)}]{\(e_{t_5a}\)}};
        \node[element, text = magenta!80] (5y) at (0, -2) {\(\vphantom{Ae_y}e_{t_5y}\)};
        \node[element, text = magenta!80] (5c) at (1, -2) {\(\vphantom{Ae_y}\)\makebox[\widthof{\(e_{t_1w}\)}]{\(e_{t_5c}\)}};
        \node [draw = none, text = orange, left= 15pt of 1u, minimum width = 5pt] {\(L_{t_1}:\)};
        \node [draw = none, text = green!70!black, left= 15pt of 2x, minimum width = 5pt] {\(L_{t_2}:\)};
        \node [draw = none, text = blue!80, left= 15pt of 3x, minimum width = 5pt] {\(L_{t_3}:\)};
        \node [draw = none, text = red!80, left= 15pt of 4x, minimum width = 5pt] {\(L_{t_4}:\)};
        \node [draw = none, text = magenta!80, left= 15pt of 5a, minimum width = 5pt] {\(L_{t_5}:\)};

        \node [draw = none, text = black!60, minimum width = 5pt] (Pu) at (-0.7, 2.3) {\(\vphantom{P_y}P_u\)};
        \node [draw = none, text = black!60, minimum width = 5pt] (Pv) at (0.3, 2.3) {\(\vphantom{P_y}P_v\)};
        \node [draw = none, text = black!60, minimum width = 5pt] (Pw) at (1.3, 2.3) {\(\vphantom{P_y}P_w\)};
        \node [draw = none, text = black!60, minimum width = 5pt] (Px) at (-0.7, 1.3) {\(\vphantom{P_y}P_x\)};
        \node [draw = none, text = black!60, minimum width = 5pt] (Pb) at (0.3, 0.3) {\(\vphantom{P_y}P_b\)};
        \node [draw = none, text = black!60, minimum width = 5pt] (Pz) at (1.3, 0.3) {\(\vphantom{P_y}P_z\)};
        \node [draw = none, text = black!60, minimum width = 5pt] (Pa) at (-0.7, -1.7) {\(\vphantom{P_y}P_a\)};
        \node [draw = none, text = black!60, minimum width = 5pt] (Py) at (0.3, -0.7) {\(\vphantom{P_y}P_y\)};
        \node [draw = none, text = black!60, minimum width = 5pt] (Pc) at (1.3, -1.7) {\(\vphantom{P_y}P_c\)};

        \begin{pgfonlayer}{background}
        \node [fit=(1u), draw = black!20, fill=black!5, inner sep=0pt] {};
        \node [fit=(1v) (2v), draw = black!20, fill=black!5, inner sep=0pt] {};
        \node [fit=(1w) (2w), draw = black!20, fill=black!5, inner sep=0pt] {};
        \node [fit=(2x) (3x) (4x), draw = black!20, fill=black!5, inner sep=0pt] {};
        \node [fit=(3b), draw = black!20, fill=black!5, inner sep=0pt] {};
        \node [fit=(3z) (4z), draw = black!20, fill=black!5, inner sep=0pt] {};
        \node [fit=(4y) (5y), draw = black!20, fill=black!5, inner sep=0pt] {};
        \node [fit=(5a), draw = black!20, fill=black!5, inner sep=0pt] {};
        \node [fit=(5c), draw = black!20, fill=black!5, inner sep=0pt] {};
        \end{pgfonlayer}
        \end{tikzpicture}
        \end{minipage}
\caption{Illustration of the reduction from 3DM to Matroid-FRA with two voters. The left figure shows an instance of 3DM, and the right figure shows the partition matroid in the constructed Matroid-FRA instance. In the two voters' preference lists, the blocks $L_t$ appear in reverse order, while the order of elements within each block is the one shown in the figure for both voters.}
\label{fig:two}
\end{figure}

Given an instance $I$ of $3$DM as described above, we construct an instance $I'$ of Matroid-FRA with a target value $\theta$ as follows. (See Figure~\ref{fig:two}.) For each triple $t=(x,y,z)\in T$, we create three elements $e_{tx}, e_{ty}, e_{tz}$. Thus, the ground set in $I'$ is $E=\bigcup_{t=(x,y,z)\in T}\{e_{tx}, e_{ty}, e_{tz}\}$, and $|E|=3|T|$. We next define a matroid. For each $v\in X\cup Y\cup Z$, let
\[
P_v\coloneqq \set{e_{tv}\in E | v\in t\in T}.
\]
Then, the family $\set{P_v | v\in X\cup Y\cup Z}$ forms a partition of $E$. The matroid $M=(E,\B)$ is defined as the partition matroid with respect to this partition, i.e.,
\[
\B=\set{B\subseteq E | \ |B\cap P_v|=1~\text{for all}~v\in X\cup Y\cup Z}.
\]
This matroid has rank $k=3q$.

We next define the two rankings $\pi_1$ and $\pi_2$ over $E$.
For each $t\in T$, let $L_t$ be the ordered list $(e_{tx}, e_{ty}, e_{tz})$ of the three elements generated from $t$. The underlying sets of these lists form a partition of $E$.
Fix an arbitrary total order $\rho$ over $T$ and define the rankings $\pi_1$ and $\pi_2$ over $E$ by concatenating these short lists as follows:
\begin{align*}
\pi_1&=L_{\rho(1)}L_{\rho(2)}\cdots L_{\rho(|T|)},\\
\pi_2&=L_{\rho(|T|)}L_{\rho(|T|-1)}\cdots L_{\rho(1)}.
\end{align*}
Set a target objective value $\theta=\binom{3|T|}{2}-3|T|$. This completes the construction of $I'$. We first make the following observation.

\begin{claim}\label{claim:2voter}
For any (possibly infeasible) ranking $\sigma$ over $E$, its objective value $\sum_{i=1}^2 d_{\mathrm{KT}}(\sigma,\pi_i)$ is at least $\theta$. It is exactly $\theta$ if and only if $e_{tx}\prec_{\sigma}e_{ty}\prec_{\sigma}e_{tz}$ for every $t=(x,y,z)\in T$.
\end{claim}
\begin{proof}
For any ranking $\sigma$ over $E$ and any pair $\{e,f\}\subseteq E$, let $\Delta_{\sigma}(\{e,f\})$ be the number of $\pi_i$ $(i\in\{1,2\})$ such that $\pi_i$ and $\sigma$ are inconsistent on $\{e,f\}$. Then $\Delta_{\sigma}(\{e,f\})\in \{0,1,2\}$. Using this, the objective value can be written as
\[
\sum_{i=1}^2 d_{\mathrm{KT}}(\sigma,\pi_i)=\sum_{\{e,f\}\subseteq E}\Delta_{\sigma}(\{e,f\}).
\]

For any $\{e,f\}\subseteq E$, we have $\Delta_{\sigma}(\{e,f\})\in\{0,2\}$ if $\pi_1$ and $\pi_2$ are consistent on $\{e,f\}$, and otherwise $\Delta_{\sigma}(\{e,f\})=1$. Observe that $\pi_1$ and $\pi_2$ are consistent on $\{e,f\}$ if and only if $e$ and $f$ belong to the same list $L_t$. Therefore, $3|T|$ pairs take values in $\{0,2\}$, while $\binom{|E|}{2}-3|T|=\theta$ pairs take value $1$. This immediately implies that the objective value is at least $\theta$.

The objective value is exactly $\theta$ if and only if $\Delta_{\sigma}(\{e,f\})=0$ for every pair $\{e,f\}$ belonging to the same short list $L_t$. This condition is equivalent to the condition that for each $t=(x,y,z)\in T$, the three elements $e_{tx}, e_{ty}, e_{tz}$ appear in this order in $\sigma$.
\end{proof}

Using this claim, we show both directions.
We first show that if $I$ admits a perfect $3$-dimensional matching $N$, then there is a feasible ranking with objective value $\theta$. Suppose that such an $N$ exists, and define $B_N=\bigcup_{t=(x,y,z)\in N}\{e_{tx}, e_{ty}, e_{tz}\}$.
Then $|B_N\cap P_v|=1$ for every $v\in X\cup Y\cup Z$, and hence $B_N\in \B$.
Define a ranking $\sigma_N$ over $E$ so that (1) the top $3q$ elements consist of $B_N$, obtained by concatenating the $q$ short lists $L_t$ for $t\in N$ in an arbitrary order, and (2) the remaining bottom part is obtained by concatenating the $|T|-q$ short lists $L_t$ for $t\in T\setminus N$ in an arbitrary order. Then, for every $t=(x,y,z)\in T$, it satisfies $e_{tx}\prec_{\sigma_N}e_{ty}\prec_{\sigma_N}e_{tz}$. This ranking $\sigma_N$ is feasible because $\Top_k(\sigma_N)=B_N\in \B$ and attains the objective value $\theta$ by Claim~\ref{claim:2voter}.

We next show the other direction. Suppose that $I'$ admits a feasible ranking $\sigma$ whose objective value is $\theta$. For convenience, write $Z=\{z_1,z_2,\dots,z_q\}$.
Since $\Top_k(\sigma)\in \B$, for each $z_i$, we have $|\Top_k(\sigma)\cap P_{z_i}|=1$. Let $t_i\in T$ be such that $\Top_k(\sigma)\cap P_{z_i}=\{e_{t_i z_i}\}$, and let $x_i$ and $y_i$ be its $X$- and $Y$-elements, i.e., $t_i=(x_i,y_i,z_i)\in T$. Then $t_1,t_2,\dots,t_q$ are all distinct, since their $Z$-elements are distinct (while the $x_i$'s and $y_i$'s are not guaranteed to be distinct at this point). Set $N_\sigma=\set{t_i=(x_i,y_i,z_i)| i\in\{1,2,\dots,q\}}$.
We now show that $N_\sigma$ is a perfect $3$-dimensional matching in $I$.
For each $i\in \{1,2,\dots,q\}$, Claim~\ref{claim:2voter} implies
$e_{t_i x_i}\prec_{\sigma}e_{t_i y_i}\prec_{\sigma}e_{t_i z_i}$,
and hence $e_{t_i z_i}\in \Top_k(\sigma)$ implies $e_{t_i x_i}, e_{t_i y_i}\in \Top_k(\sigma)$. Since $k=3q$ and all $t_i$ are distinct, $\Top_k(\sigma)$ coincides with
$\bigcup_{i=1}^q\{e_{t_i x_i},e_{t_i y_i},e_{t_i z_i}\}$.
The condition $\Top_k(\sigma)\in \B$ then forces
$\textstyle{\left|\bigcup_{i=1}^q\{e_{t_i x_i},e_{t_i y_i},e_{t_i z_i}\}\cap P_v\right|=1}$ for each $v\in X\cup Y\cup Z$,
implying that $x_1,\dots,x_q$, $y_1,\dots,y_q$, and $z_1,\dots,z_q$ are all distinct. This shows that $N_{\sigma}$ is a perfect $3$-dimensional matching.
\end{proof}

\medskip

\ThmHardThree*

To obtain the theorem, we show NP-completeness of the following decision version.

\begin{theorem}
It is NP-complete to determine, given a number $\theta$ and an instance $I$ of Matroid-FRA with $m=3$ and a partition matroid specified by its partition and capacities, whether $I$ admits a feasible ranking with objective value at most $\theta$.
\end{theorem}

\begin{proof}
Clearly, the problem is in NP. We show NP-hardness by a reduction from Exact Cover by $3$-Sets (X3C), a classical NP-complete problem~\cite{garey1979}. An instance of X3C consists of a finite set $U$ and a family $\C=\{C_1, \dots,C_q\} \subseteq 2^U$ with $|C_j|=3$ for each $j=1,\dots, q$. The question is whether there exists a subfamily $\X\subseteq \C$ such that each $u\in U$ is contained in exactly one set of $\X$. We may also assume that every element of $U$ belongs to at least one set in $\C$, since otherwise the instance is trivially a no-instance.

Given an instance $I$ of X3C as described above, we construct an instance $I'$ of Matroid-FRA with a target value $\theta$ as follows. Without loss of generality, we may assume that $q=2^h$ for some integer $h$, since otherwise we can add copies of existing $C_j$ to $\C$ (so that $\C$ becomes a multiset) to satisfy this condition; the size of the resulting $\C$ is at most twice the original size. Thus, $h=\log_2 q$. Set $d\coloneqq h+2$ and let $E=\{0,1,2\}^d$. Then $|E|=3^d=9\cdot q^{\log_2 3}$, which is polynomial in the input size. To define the three rankings $\pi_1, \pi_2, \pi_3$ over $E$, we first define the underlying rankings $\rho_1,\rho_2,\rho_3$ over $\{0,1,2\}$ as follows:
\[
0 \prec_{\rho_1} 1 \prec_{\rho_1} 2,\qquad 1 \prec_{\rho_2} 2 \prec_{\rho_2} 0, \qquad 2 \prec_{\rho_3} 0 \prec_{\rho_3} 1.
\]
For each $i=1,2,3$, the ranking $\pi_i$ is defined as the lexicographic order with respect to $\rho_i$. To be more precise, for any $e\in E=\{0,1,2\}^d$, let $e_j$ denote its $j$th digit. Then,
\[
e\prec_{\pi_i}f ~~\Longleftrightarrow~~ \text{there exists }\ell\in [d]\text{ such that } e_j=f_j \text{ for all }j\in[\ell-1]\text{ and } e_\ell\prec_{\rho_i}f_\ell.
\]

Before specifying the matroid, we analyze the behavior of the objective value.
We use the notation
$\{0,1,2\}^{<d}\coloneqq \{\varepsilon\} \cup \{0,1,2\}^1\cup \{0,1,2\}^2\cup \cdots \cup \{0,1,2\}^{d-1}$,
where $\varepsilon$ denotes the empty sequence.
For each $\alpha\in \{0,1,2\}^{<d}$, we write
$\llbracket \alpha \rrbracket\coloneqq\set{\beta\in \{0,1,2\}^d | \text{$\beta$ has $\alpha$ as a prefix}}$,
where $\varepsilon$ is a prefix of every element of $E$.

\medskip
\noindent{{\bf Step\,1. Analyzing the objective function.}}
Let $\sigma$ be any ranking over $E$.
Similarly to the proof of Claim~\ref{claim:2voter}, for any pair $\{e,f\}\subseteq E$, let $\Delta_{\sigma}(\{e,f\})$ be the number of $\pi_i$ $(i\in\{1,2,3\})$ such that $\pi_i$ and $\sigma$ are inconsistent on $\{e,f\}$.

We partition the set $\binom{E}{2}$ of pairs into classes according to the longest common prefix. For each $\alpha\in \{0,1,2\}^{<d}$, let
\[
\mathcal{P}_\alpha=\set{\{e,f\}\subseteq E | \text{$\alpha$ is the longest common prefix of $e$ and $f$}}.
\]
That is, $\mathcal P_\alpha$ consists of all unordered pairs
$\{e,f\}$ such that, after possibly swapping $e$ and $f$, the ordered pair
\((e,f)\) belongs to $(\llbracket \alpha 0 \rrbracket\times \llbracket \alpha 1 \rrbracket)\cup (\llbracket \alpha 1 \rrbracket\times \llbracket \alpha 2 \rrbracket)\cup (\llbracket \alpha 2 \rrbracket\times \llbracket \alpha 0 \rrbracket)$.

By definition, $\{\mathcal{P}_\alpha\}_{\alpha\in \{0,1,2\}^{<d}}$ is a partition of $\binom{E}{2}$.
For each $\alpha\in \{0,1,2\}^{<d}$, define
\[
H_{\alpha}(\sigma)=\sum_{\{e,f\}\in \mathcal{P}_\alpha}\Delta_{\sigma}(\{e,f\}).
\]
Thus, $H_\alpha(\sigma)$ is the contribution to the objective value from the pairs in $\mathcal{P}_\alpha$.

For each $\alpha\in \{0,1,2\}^{<d}$, let $|\alpha|$ denote the length (i.e., the number of digits) of $\alpha$, and let $n_{\alpha}\coloneqq 3^{d-|\alpha|-1}$, which is the size of each of the sets $\llbracket \alpha 0 \rrbracket$, $\llbracket \alpha 1 \rrbracket$, and $\llbracket \alpha 2 \rrbracket$.
\begin{claim}\label{claim:3voter}
For any $\alpha\in \{0,1,2\}^{<d}$, the following hold.
\begin{enumerate}
\item We always have $H_{\alpha}(\sigma)\geq 4\cdot n_\alpha^2$.
\item The equality holds if and only if, for every triple $(e,f,g)\in \llbracket \alpha 0 \rrbracket\times \llbracket \alpha 1 \rrbracket\times \llbracket \alpha 2 \rrbracket$, the elements $e,f,g$ are ordered in $\sigma$ in this order or in a cyclic shift of this order.
\item If the equality holds, then for any $\ell\in[n]$, one of the following must hold:
\begin{itemize}
\item $\llbracket \alpha 0 \rrbracket$ or $\llbracket \alpha 1 \rrbracket$ or $\llbracket \alpha 2 \rrbracket$ is completely contained in $\Top_\ell(\sigma)$, or
\item $\llbracket \alpha 0 \rrbracket$ or $\llbracket \alpha 1 \rrbracket$ or $\llbracket \alpha 2 \rrbracket$ is disjoint from $\Top_\ell(\sigma)$.
\end{itemize}
\end{enumerate}
\end{claim}
\begin{proof}
Note that each pair $\{e,f\}\in \mathcal{P}_\alpha$ consists of elements taken from two distinct sets among
$\llbracket \alpha 0 \rrbracket$, $\llbracket \alpha 1 \rrbracket$, and $\llbracket \alpha 2 \rrbracket$, each of which has size $n_{\alpha}$.
By the definition of $\pi_i$ $(i=1,2,3)$, which are the lexicographic orders with respect to $\rho_i$, we have
\begin{equation}
(e,f)\in (\llbracket \alpha 0 \rrbracket\times \llbracket \alpha 1 \rrbracket) \cup (\llbracket \alpha 1 \rrbracket\times \llbracket \alpha 2 \rrbracket)\cup (\llbracket \alpha 2 \rrbracket\times \llbracket \alpha 0 \rrbracket)\implies
\Delta_{\sigma}(\{e,f\})=
\begin{cases}
1 & \text{if } e\prec_{\sigma} f,\\
2 & \text{if } f\prec_{\sigma} e.
\end{cases}
\label{onetwo}
\end{equation}
We next show the following equality:
\[
\sum_{\{e,f\}\in \mathcal{P}_\alpha}\Delta_{\sigma}(\{e,f\})=\frac{1}{n_{\alpha}}\cdot \sum_{(e,f,g)\in \llbracket \alpha 0 \rrbracket\times \llbracket \alpha 1 \rrbracket\times \llbracket \alpha 2 \rrbracket}\left(\Delta_{\sigma}(\{e,f\})+\Delta_{\sigma}(\{f,g\})+\Delta_{\sigma}(\{g,e\})\right).
\]
To see this, fix any pair in $\mathcal{P}_\alpha$, say $\{e^*,f^*\}\in \mathcal{P}_{\alpha}$. Assume without loss of generality that $(e^*,f^*)\in \llbracket \alpha 0 \rrbracket\times \llbracket \alpha 1 \rrbracket$. Then on the right-hand side, the term $\Delta_{\sigma}(\{e^*,f^*\})$ appears $n_\alpha$ times as $(e^*,f^*, g)$ with different $g\in \llbracket \alpha 2 \rrbracket$. Thus this equality holds.

Take any $(e,f,g)\in \llbracket \alpha 0 \rrbracket\times \llbracket \alpha 1 \rrbracket\times \llbracket \alpha 2 \rrbracket$. By \eqref{onetwo}, each term of
$\Delta_{\sigma}(\{e,f\})+\Delta_{\sigma}(\{f,g\})+\Delta_{\sigma}(\{g,e\})$ is $1$ or $2$, but it is impossible that all of them are $1$, since otherwise we would have $e\prec_\sigma f\prec_\sigma g\prec_\sigma e$, a contradiction. Thus, the minimum value of the sum of these three terms is $4$, and it is attained if and only if we have either
$e\prec_\sigma f\prec_\sigma g$, ~$f\prec_\sigma g\prec_\sigma e$, or $
g\prec_\sigma e\prec_\sigma f$,
i.e., $e,f,g$ appear in $\sigma$ in this order or in a cyclic shift of this order. Then
\[
H_{\alpha}(\sigma)=\sum_{\{e,f\}\in \mathcal{P}_\alpha}\Delta_{\sigma}(\{e,f\})
\geq \frac{1}{n_{\alpha}}\cdot n_\alpha^3 \cdot 4
=4\cdot n_{\alpha}^2
=4\cdot 3^{2(d-|\alpha|-1)},
\]
and the equality holds if and only if for every triple $(e,f,g)\in \llbracket \alpha 0 \rrbracket\times \llbracket \alpha 1 \rrbracket\times \llbracket \alpha 2 \rrbracket$, the elements $e,f,g$ are ordered in $\sigma$ in this order or in a cyclic shift of this order. Thus, the first and second statements of the claim are shown.

To prove the third statement, suppose for the sake of contradiction that we have $H_\alpha(\sigma)=4\cdot 3^{2(d-|\alpha|-1)}$ while none of $\llbracket \alpha 0 \rrbracket$, $\llbracket \alpha 1 \rrbracket$, and $\llbracket \alpha 2 \rrbracket$ is contained in or disjoint from $\Top_\ell(\sigma)$.
Then, there exist $e_1, e_2\in \llbracket \alpha 0 \rrbracket$, $f_1, f_2\in \llbracket \alpha 1 \rrbracket$, and $g_1, g_2\in \llbracket \alpha 2 \rrbracket$ such that $e_1, f_1, g_1\in \Top_\ell(\sigma)$ and $e_2, f_2, g_2\in E\setminus \Top_\ell(\sigma)$.

Consider the triple $(e_1, f_1, g_2)$. By the above discussion, for $H_{\alpha}(\sigma)$ to be $4\cdot n_\alpha^2$, these three elements must appear in this order or in a cyclic shift of it. However, since $g_2\in E\setminus \Top_\ell(\sigma)$ cannot precede $e_1,f_1\in  \Top_\ell(\sigma)$, we must have $e_1\prec_\sigma f_1\prec_\sigma g_2$.
Similar arguments applied to the triples $(e_2,f_1, g_1)$ and $(e_1,f_2, g_1)$ respectively imply
$f_1\prec_\sigma g_1\prec_\sigma e_2$ and $g_1\prec_\sigma e_1\prec_\sigma f_2$.
However, these relations imply
$e_1\prec_\sigma f_1\prec_\sigma g_1 \prec_\sigma e_1$, a contradiction.
\end{proof}

By the definition of $H_\alpha(\sigma)$ and Claim~\ref{claim:3voter}, the objective value satisfies
\[
\sum_{i=1}^3 d_{\mathrm{KT}}(\sigma,\pi_i)=\sum_{\alpha\in \{0,1,2\}^{<d}}H_{\alpha}(\sigma)\geq \sum_{\alpha\in \{0,1,2\}^{<d}}4\cdot n_{\alpha}^2,
\]
for any ranking $\sigma$ over $E$. We set the target objective value in $I'$ as $\theta=\sum_{\alpha\in \{0,1,2\}^{<d}}4\cdot n_{\alpha}^2$,
i.e., the right-hand side of the above inequality.%
\footnote{
We remark that each input ranking $\pi_i$, when substituted for $\sigma$, attains this lower bound, since it satisfies the equality condition in Claim~\ref{claim:3voter} for every $\alpha$. Thus, all three input rankings are optimal for the unconstrained Kemeny problem.
Our reduction shows that $3$-voter Kemeny aggregation under
matroid constraints is NP-hard even for profiles with this property.}

\begin{figure}[t]
\scalebox{0.80}{
  \input{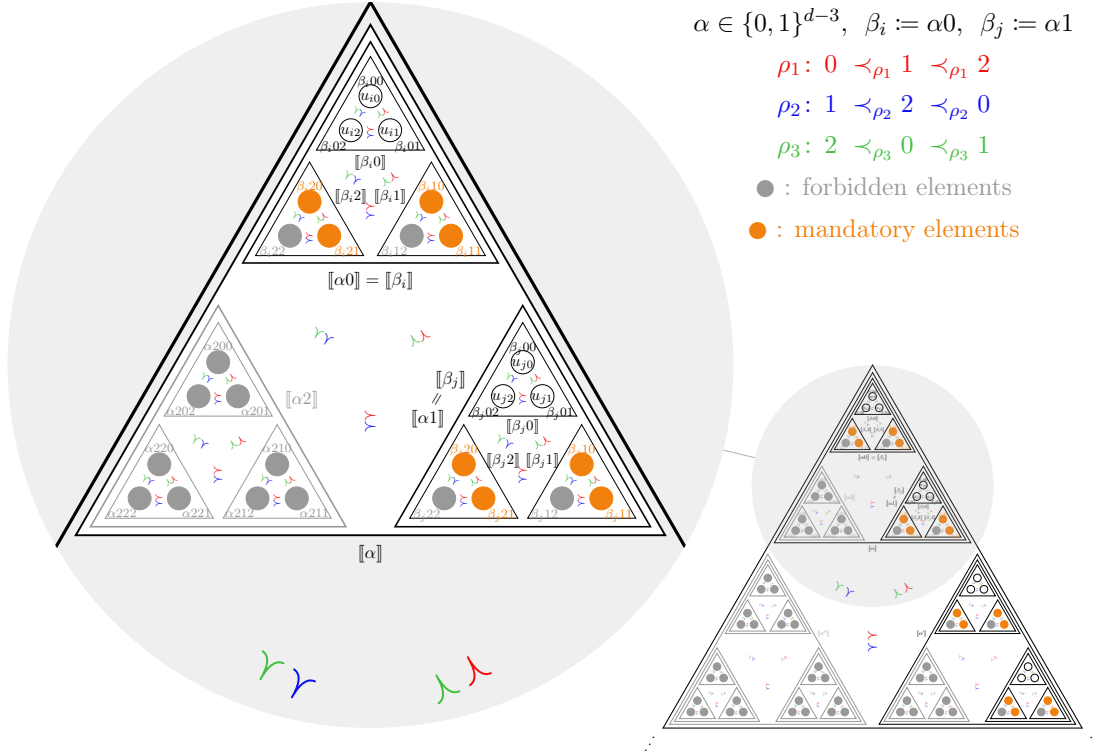}
}
\caption{An illustration of the constructed Matroid-FRA instance with three voters, obtained by a reduction from X3C. The ground set is $\{0,1,2\}^d$, and each triangle represents a set of elements sharing a common prefix. The figure on the right shows the set of elements with a common prefix in $\{0,1\}^{d-4}$, while the enlarged figure on the left zooms in on the subset corresponding to a common prefix $\alpha \in \{0,1\}^{d-3}$. The three input rankings are defined lexicographically from the underlying orders $\rho_1, \rho_2, \rho_3$, shown at the upper right of the figure. Each 3-set in the X3C instance corresponds to a smallest triangle containing three uncolored elements of $\{0,1,2\}^d$. For each element of the ground set of the X3C instance, there is a partition class consisting of the corresponding elements of $\{0,1,2\}^d$. Elements colored in gray indicate forbidden elements (loops), and those colored in orange indicate mandatory elements that must be included in every matroid base.}
\label{fig:three}
\end{figure}

\medskip
\noindent{{\bf Step\,2. Defining the partition matroid.}}
To define the matroid, recall that $d=h+2$ and $2^h=q$, where $q$ is the number of size-$3$ sets in the X3C instance $I$. Let us enumerate the elements of $\{0,1\}^h$ as $\{0,1\}^h=\{\beta_1,\dots,\beta_q\}$.
Observe that for each $j=1,\dots, q$, since $|\beta_j|=h=d-2$, the set $\llbracket \beta_j \rrbracket$ is partitioned into the following three sets:
\begin{align*}
 \llbracket \beta_j 0 \rrbracket&=\{\beta_j 00, \beta_j 01, \beta_j 02\},\\
 \llbracket \beta_j 1 \rrbracket&=\{\beta_j 10, \beta_j 11, \beta_j 12\},\\
 \llbracket \beta_j 2 \rrbracket&=\{\beta_j 20, \beta_j 21, \beta_j 22\}.
\end{align*}
We use the sets $\llbracket \beta_j 0 \rrbracket$ to implement the size-$3$ sets $C_j=\{u_{j,0}, u_{j,1}, u_{j,2}\}$. (See Figure~\ref{fig:three}.)
For each element $u\in U$ in $I$, define a partition class
\[
P_u=\set{\beta_j 0\ell | j\in [q],\, \ell\in\{0,1,2\},\, u=u_{j,\ell}},
\]
and assign capacity $1$ to it.
In contrast, the elements in $\llbracket \beta_j 1 \rrbracket$ and $\llbracket \beta_j 2 \rrbracket$ are all used as forbidden or mandatory elements as follows. For each of $\beta_j 10, \beta_j 11, \beta_j 20, \beta_j 21$, we consider a partition class consisting only of that single element and having capacity $1$, which forces any base to include these elements; i.e., these are {\em mandatory}. We let $\beta_j 12$ and $\beta_j 22$ be loops (dependent singletons) in the matroid, meaning that these are {\em forbidden} elements.
Thus, for every base, neither $\llbracket \beta_j 1 \rrbracket$ nor $\llbracket \beta_j 2 \rrbracket$ is contained in or disjoint from the base.

We have thus specified the roles of the elements in $\llbracket \beta \rrbracket$ for all $\beta\in \{0,1\}^h=\{0,1\}^{d-2}$. For each remaining element in $E=\{0,1,2\}^d$, i.e., each element $e\in E$ in which $2$ appears in the first $h$ digits at least once, we declare it to be a forbidden element. This completes the definition of the matroid $M$. Let $E_{\rm mand}$ and $E_{\rm forb}$ be the sets of mandatory and forbidden elements, respectively.
Then, the base family $\B$ is written as
\[
\B=\set{ B\subseteq E | |B\cap P_u|=1~(u\in U),~E_{\rm mand}\subseteq B,~E_{\rm forb}\cap B=\emptyset}.
\]
The rank of $M$ is $k=|U|+|E_{\rm mand}|$. This completes the construction of the Matroid-FRA instance $I'$.

\medskip
\noindent{{\bf Step\,3. Correctness of the reduction.}} We prove both directions in turn.

First, suppose that $I'$ admits a feasible ranking $\sigma$ with objective value $\theta$. We show that the X3C instance $I$ admits an exact cover.
Since the objective value is $\theta=\sum_{\alpha\in \{0,1,2\}^{<d}}4\cdot n_{\alpha}^2$, the third statement of Claim~\ref{claim:3voter} implies that, for each $\alpha\in \{0,1,2\}^{<d}$, at least one of $\llbracket \alpha 0 \rrbracket$, $\llbracket \alpha 1 \rrbracket$, and $\llbracket \alpha 2 \rrbracket$ is completely contained in or disjoint from $\Top_k(\sigma)$. In particular, when applied to $\beta_j$ $(j=1,\dots, q)$, this implies that $\llbracket \beta_j 0 \rrbracket$ is either contained in or disjoint from $\Top_k(\sigma)$, because feasibility of $\sigma$ implies that $\Top_k(\sigma)$ includes $\beta_j 10, \beta_j 11, \beta_j 20, \beta_j 21$ and excludes $\beta_j 12$ and $\beta_j 22$. Thus, we have $|\llbracket \beta_j 0 \rrbracket\cap \Top_k(\sigma)|\in\{0,3\}$
for each $j=1,\dots, q$. Let
$\X_{\sigma}=\set{C_j| \llbracket \beta_j 0 \rrbracket\subseteq \Top_k(\sigma)}$.
Since $\Top_k(\sigma)$ is a base, for each $u\in U$, we have $|P_u\cap \Top_k(\sigma)|=1$, which means that exactly one set $C_j$ in $\X_{\sigma}$ contains $u$. Thus, $\X_{\sigma}$ is an exact cover in $I$.

We next show the other direction.
Suppose that the X3C instance $I$ admits an exact cover $\X\subseteq \C$. We show that $I'$ admits a feasible ranking $\sigma$ with objective value $\theta$.
Let $B_{\X}=\left(\bigcup_{j:C_j\in \X}\llbracket \beta_j 0 \rrbracket\right)\cup E_{\rm mand}$.
Since $\X$ is an exact cover, we have $|B_{\X}\cap P_u|=1$ for every $u\in U$, and hence $B_{\X}\in \B$. We now construct a ranking $\hat{\sigma}$ satisfying $\Top_k(\hat{\sigma})=B_{\X}$.
We define the order $\hat{\sigma}$ recursively.
\begin{enumerate}
\item For each $\alpha\in \{0,1,2\}^d$, let $\hat{\sigma}_\alpha$ be the trivial ranking on the singleton $\{\alpha\}$.
\item For each $\alpha\in \{0,1,2\}^{<d}$, we define $\hat{\sigma}_{\alpha}$ using $\hat{\sigma}_{\alpha 0}$, $\hat{\sigma}_{\alpha 1}$, and $\hat{\sigma}_{\alpha 2}$ as follows.
First, partition the set $\llbracket \alpha \rrbracket$ into the following subsets:
\begin{align*}
E_1&=\llbracket \alpha 0 \rrbracket\cap B_{\X},\quad
F_1=\llbracket \alpha 1 \rrbracket\cap B_{\X},\quad
G_1=\llbracket \alpha 2 \rrbracket\cap B_{\X},\\
E_2&=\llbracket \alpha 0 \rrbracket\setminus B_{\X},\quad
F_2=\llbracket \alpha 1 \rrbracket\setminus B_{\X},\quad
G_2=\llbracket \alpha 2 \rrbracket\setminus B_{\X}.
\end{align*}
By the definition of $B_{\X}$ and the set $E_{\rm mand}$, at least one of $E_1=\emptyset$ or $E_2=\emptyset$ or $G_1=\emptyset$ holds. (The first or second holds when $|\alpha|\geq d-2$, and otherwise the third holds.)

Define a ranking $\hat{\sigma}_{\alpha}$ on $\llbracket \alpha \rrbracket$ as follows.
\begin{itemize}
\item If $E_1=\emptyset$, let $\hat{\sigma}_{\alpha}$ arrange the blocks in the order $F_1\, G_1\, G_2\, E_2\, F_2$.
\item If $E_2=\emptyset$, let $\hat{\sigma}_{\alpha}$ arrange the blocks in the order $G_1\, E_1\, F_1\, F_2 \, G_2$.
\item If $G_1=\emptyset$, let $\hat{\sigma}_{\alpha}$ arrange the blocks in the order $E_1\, F_1\, F_2\, G_2\, E_2$.
\end{itemize}
Here, the ordering inside $E_1\cup E_2$ (resp., $F_1\cup F_2$, $G_1\cup G_2$) is given by the orderings $\hat{\sigma}_{\alpha 0}$ (resp., $\hat{\sigma}_{\alpha 1}$, $\hat{\sigma}_{\alpha 2}$).
\end{enumerate}
Let $\hat{\sigma}=\hat{\sigma}_{\varepsilon}$ be the ranking on $E$ obtained in this way. In each of the three cases in the construction, all elements of
$\llbracket \alpha\rrbracket\cap B_{\X}$ are placed before all elements of
$\llbracket \alpha\rrbracket\setminus B_{\X}$. Hence, by induction on the
recursion, every element of $B_{\X}$ precedes every element of
$E\setminus B_{\X}$ in $\hat{\sigma}$. Therefore $\Top_k(\hat{\sigma})=B_{\X}$.

By construction, for every $\alpha\in \{0,1,2\}^{<d}$, the restriction of $\hat{\sigma}$ to $\llbracket \alpha \rrbracket$ coincides with $\hat{\sigma}_{\alpha}$.
The construction also ensures that every $(e,f,g)\in \llbracket \alpha 0 \rrbracket\times \llbracket \alpha 1 \rrbracket\times \llbracket \alpha 2 \rrbracket$ is ordered in $\hat{\sigma}_{\alpha}$ in this order or in a cyclic shift of this order (note that $e\in E_1\cup E_2$, $f\in F_1\cup F_2$, and $g\in G_1\cup G_2$ in the above construction). Therefore, by the second statement of Claim~\ref{claim:3voter}, we have $H_{\alpha}(\hat{\sigma})=4\cdot n_\alpha^2$ for every $\alpha\in \{0,1,2\}^{<d}$, and hence the objective value $\sum_{i=1}^3 d_{\mathrm{KT}}(\hat{\sigma},\pi_i)$ coincides with
$\sum_{\alpha\in \{0,1,2\}^{<d}}4\cdot n_{\alpha}^2=\theta$.
\end{proof}

\CorAll*
\begin{proof}
If $m$ is even, we start from the NP-hard case $m=2$; if $m$ is odd, we start from the NP-hard case $m=3$. To increase the number of voters by $2$, add an arbitrary ranking $\pi$ and its reverse $\pi^{\mathrm{rev}}$ on the same ground set. For every ranking $\sigma$,
$d_{\mathrm{KT}}(\sigma,\pi)+d_{\mathrm{KT}}(\sigma,\pi^{\mathrm{rev}})
= \binom{|E|}{2}$,
which is independent of $\sigma$. Hence each added pair shifts the objective value by the same constant and does not affect which feasible rankings are optimal. Repeating this extension yields NP-hardness for every $m>1$.
\end{proof}

\end{document}